%% file: IJK-Gr-revised_for_JGP.tex
\documentclass{article} 
\usepackage{graphicx,amsmath, amssymb, verbatim}
\usepackage{framed}
\usepackage{fancybox}
\usepackage{amsthm}
\usepackage[all]{xy}
\usepackage{bm}
\usepackage{ascmac}
\usepackage{physics}
\usepackage{color}
\usepackage{comment}
\allowdisplaybreaks[2]
\newtheorem{proposition}{Proposition}
\newtheorem{theorem}{Theorem}
\newtheorem{definition}{Definition}
\newtheorem*{definition*}{Definition}

\newtheorem{remark}{Remark}
\newtheorem{lemma}{Lemma}

\newcommand {\R } {\mathbb{R}}

\newcommand{\ba}{\begin{eqnarray}}
\newcommand{\ea}{\end{eqnarray}}
\newcommand{\no}{\nonumber}

\newcommand{\bt}{\beta}

\newcommand{\mul}{\widetilde{\mathrm{mul}}}
\newcommand{\lag}{\langle}
\newcommand{\rag}{\rangle}
\def\d{{\partial}}

\begin{document} 

\title{\bf{Evaluation of Euler Number of Complex Grassmann Manifold  
$G(k, N)$ via Mathai-Quillen Formalism}} 
\author{\large Shoichiro Imanishi${}^{(1)}$,\; Masao Jinzenji${}^{(2)}$, Ken Kuwata${}^{(3)}$ \\
\\ 
${}^{(1)}$\it Division of Mathematics, Graduate School of Science \\
\it Hokkaido University \\
\it  Kita-ku, Sapporo, 060-0810, Japan\\
\\
${}^{(2)}$ \it Department of Mathematics,  \\
\it Okayama University \\
\it  Okayama, 700-8530, Japan\\
\\
\\ 
${}^{(3)}$\it Department of General Education \\
\it National Institute of Technology, Kagawa Colledge \\
\it  Chokushi, Takamatsu, 761-8058, Japan\\
\\
\it e-mail address: ${}^{(1)}$ shoichiro.imanishi@gmail.com \\ 
\it\hspace{2.2cm}${}^{(2)}$ pcj70e4e@okayama-u.ac.jp \\
\it\hspace{2.2cm}${}^{(3)}$ kuwata@t.kagawa-nct.ac.jp }

\maketitle 

\begin {abstract} 
In this paper, we provide a recipe for computing Euler number of Grassmann manifold $G(k,N)$ by using Mathai-Quillen formalism (MQ formalism)\cite{MQ1} and Atiyah-Jeffrey construction \cite{AJ1}. Especially, we construct path-integral representation of Euler number of $G(k,N)$. Our model corresponds to a finite dimensional toy-model of topological Yang-Mills theory which motivated Atiyah-Jeffrey construction. As a by-product, we construct free fermion realization of cohomology ring of $G(k,N)$.
\end {abstract}

\section{Introduction}
 Our aim of this paper is to compute Euler number of  finite dimensional manifold (Grassmann manifold $G(k,N)$) by using Mathai-Quillen formalism (MQ formalism) and Atiyah-Jeffrey construction  \cite{AJ1}. 
 Mathai-Quillen formalism \cite{MQ1} is a method for constructing Thom class of finite-dimensional vector bundle $E$ on a manifold 
$M$, that decreases like Gaussian along fiber direction.
This Thom class plays the same role as the original Thom class which has compact support along fiber direction. 
The Euler class of $E$ is given by pull-back of the Thom class by the section $s:M\rightarrow E$. It does not depend on choice of the section of $s$ as a cohomology class
because of homotopy invariance of de Rham cohomology. Therefore, we can compute the Euler number of $E$ by choosing a convenient section, which leads us to localization technique. 
In this paper, we choose zero-section as $s$. 
Suppose $M$ is given as an orbit space $X/G$ where a Lie group $G$ acts freely on a manifold $X$.  Atiyah and Jefferey extended the MQ-formalism for an orbit space $X/G$ \cite{AJ1}.
Strictly speaking, they extended it to tha case when $X/G$ is given by infinite-dimensional space of gauge equivalence classes of connections of $SU(2)$ bundle on a 4-dimensional manifold, in order to study mechanism behind Witten's  construction of topological Yang-Mills theory. 
We call this method ``Atiyah-Jeffrey construction''. In this paper, we apply Atiyah-Jeffrey construction to the case when $X/G$ is finite-dimensional complex Grassmann manifold $G(k,N)$ and $E$ is holomorphic tangent bundle of $G(k,N)$, 
and we construct path-integral representation of Euler number of $G(k,N)$. Our construction corresponds to a finite dimensional toy-model of topological Yang-Mills theory which motivated Atiyah-Jeffrey construction.     
\subsection{Our Model and Main Theorem}
 Complex Grassmann manifold $G(k,N)$ is a space which parametrizes $k$-dimensional linear subspaces of $N$-dimensional complex vector space. We denote by $U(k)$ unitary group that acts on  complex $k$-dimensional vector space and by $V_k(\mathbb{C}^N)$  Stiefel manifold of orthonormal $k$-frames in $\mathbb{C}^N$. Then $G(k,N)\simeq V_k(\mathbb{C}^N)/U(k)$. We introduce detailed 
informations of $G(k,N)$ in Section 3.  By applying Atiyah-Jeffrey construction to the case when $X$ and $G$ are given by $V_k(\mathbb{C}^N)$ and $U(k)$, Euler number of $G(k,N)$ is represented by the following finite-dimensional path-integral (unless otherwise noted, we use the Einstein convention).
\begin{theorem}{\bf (Main Theorem)}
Euler number of $G(k,N)$ is evaluated by finite dimensional path-integral:
\begin{align}
\chi(G(k,N))&=\binom{N}{k}=\chi(V_k(\mathbb{C}^N)/U(k))\notag \\
&=\bt\int_{V_k(\mathbb{C}^N)} Dz \int D\psi D\phi D\bar \phi DA D\eta D\chi DH   \,\,\omega \exp(-\mathcal{L}_{MQ}),
\end{align}
where the Lagrangian $\mathcal{L}_{MQ}$, the projection operator $\omega$ and the normalization factor $\bt$ are given by,  
\begin{align}
\mathcal{L}_{MQ}&=\delta<\chi,H>+\delta<\psi_A,A>+\frac{i}{2}\delta\Bigl[<\psi,iz{}^{t}\bar \phi>+\overline{<\psi,iz{}^{t}\bar \phi>}\Bigr] 
\label{gmq}\\
&=\sum^N_{s=1}\Bigl[\sum^k_{i=1}H^{\bar i}_{\bar s} H^i_s-\bigl(\chi^{\bar i}_{\bar s}(\gamma\delta_{\bar i j}+i\phi_{\bar i j} )\chi^j_s \bigr) +\psi^{\bar i}_{\bar s} \bar \phi _{\bar i l} \psi^{ l}_s \Bigr]+i\mathrm{tr}(\phi \bar \phi) \notag \\
&+\frac{1}{2}\sum_{s=1}^N \left\{ \psi^{\bar i}_{\bar s} z^{ j}_s-\psi^{ j}_{ s}  z^{\bar  i}_{\bar s}\right\}\eta_{\bar i j} +\mathrm{tr}(A^\dag(\gamma A+[i\phi,A])-\psi^\dag_A\psi_A ). \label{gmq2}, \\
\omega&=\prod_{i,j=1}^k \Bigl[\sum_{s=1}^N(\psi^{\bar i}_s z^j_s+z^{\bar i}_s \psi^j_s) \Bigr], ~~~
\bt=\frac{\prod_{j=0}^{k-1} j!}{2^{2k}  (-\pi) ^{k^2+kN}\pi^{kN+\frac{k(k+1)}{2}}(-1)^{\frac{k}{2}(k-1)}}.
\end{align}
\end{theorem}
In the above Lagrangian, $\delta$ represents supersymmetric transformation whose detailed construction will be explained in Section 2. 
In this section, we briefly introduce our notations and supersymmetric transformation. Let 
\begin{align}
&z:=(\bm{z}^1  \cdots \bm{z}^k),\\
&\bm{z}^j:=~^t(z^j_1,z^j_2,\cdots,z^j_N)\,\,\;(j=1,2,\cdots,k),
\end{align}
be local complex coordinate system of $\mathbb{C}^{kN}$. Here, we introduce some notation for our model ( this notation is used  in this chapter, Chapters 3 and Chapter 4).
\begin{definition}\label{Def1}

Let $X$and $Y$ be matrix variables used in our Lagrangian.\\
(1) 
We represent complex conjugate of $X$ by $\overline{X}$ or $~^*X$. We also use the notation $(\overline{X})_{i}^j=:X_{\bar i}^{\bar j}$ or $(\overline{X})_{\bar i j}=:X_{i \bar j}$. In our Lagrangian, $\phi$ and $\bar \phi$ are diferent independent fields, and we use $~^*\phi$ and $~^*\bar \phi$ to represent the complex conjugate of $\phi$ and $\bar \phi$. \\
(2)
We represent transpose of  $X$ (resp.  adjoint of $X$) by $~^tX$ (resp. $X^\dag$). ($X^\dag :=~^{t *}X$).
\\
(3) 
We define inner product of matrix variables $X$and $Y$ of the same type  by $<X,Y>:=\mathrm{tr}((X^{\dag}) Y)$. 
\end{definition}
 Then $V_k(\mathbb{C}^N)$ is given by a set of points in $\mathbb{C}^{kN}$ that satisfy,
\begin{align}
\sum_{s=1}^N z^{\bar i}_{\bar s}z^j_{s} - \delta_{i j}=0 \,\;\;\;(i,j=1,2,\cdots,k),
\label{stiefel}
\end{align}
where $z^{\bar i}_{\bar s}$ represents $\overline{z^{i}_{s}}$, and $\delta_{i j}$ is the Kronecker's delta. $\psi$'s are complex grassmann variables that correspond to super-partner of $z$:
\begin{align}
&\psi:=(\bm{\psi}^1,\bm{\psi}^2,\cdots,\bm{\psi}^k),& &\bm{\psi}^j:=~^t(\psi^j_1,\psi^j_2,\cdots,\psi^j_N).&
\end{align}
$\phi$ and $\bar \phi$ are $k\times k$ Hermite matrices ($\phi=\phi^{\dag}$, $\bar \phi={\bar \phi}^{\dag}$). $i\phi$ and $i\bar\phi$ plays the role of generator of Lie algebra of the gauge group $U(k)$. Note that $\bar \phi$ is not complex conjugate of $\phi$. $\eta$ is super-partner of $\phi$. Hence $\eta$ is a grassmann and Hermite matrix ($\eta=\eta^{\dag}$).
\begin{align}
&\phi :=\left(\begin{array}{ccc}
\phi_{\bar 1 1}&\cdots &\phi_{\bar 1 k}\\
\vdots&\ddots& \vdots \\
\phi_{\bar k 1}& \cdots & \phi_{\bar k k}
\end{array} \right).&
&\bar \phi :=\left(\begin{array}{ccc}
\bar \phi_{\bar 1 1}&\cdots &\bar \phi_{\bar 1 k}\\
\vdots&\ddots& \vdots \\
\bar \phi_{\bar k 1}& \cdots &\bar \phi_{\bar k k}
\end{array} \right).& \\
&\eta:=\left(\begin{array}{ccc}
\eta_{\bar 1 1}&\cdots &\eta_{\bar 1 k}\\
\vdots&\ddots& \vdots \\
\eta_{\bar k 1}& \cdots &\eta_{\bar k k}
\end{array} \right).& 
\end{align}
$H$ is $k\times N$ complex matrix which plays the role of auxiliary variable in MQ-formalism. $\chi$ is super-partner of $H$.  $A$ is $k\times k$ complex matrix and $\psi_A$ is super-partner of $A$.
\begin{align}
&H:=\left(\begin{array}{ccc}
\bm{H}^1& \cdots & \bm{H}^k
\end{array}
\right), &
&\bm{H}^j:=~^t(H^j_1,\cdots,H^j_N).& \\
&\chi:=\left(\begin{array}{ccc}
\bm{\chi}^1& \cdots & \bm{\chi}^k
\end{array}
\right), &
&\bm{\chi}^j:=~^t(\chi^j_1,\cdots,\chi^j_N).&\\
&A:=
\left(\begin{array}{ccc}
A_{\bar 1 1}&\cdots &A_{\bar 1 k}\\
\vdots&\ddots& \vdots\\
A_{\bar k1}& \cdots &A_{\bar k k}
\end{array}
\right),& 
&\psi_A:=
\left(\begin{array}{ccc}
{\psi_{A}}^1_{1}&\cdots &{\psi_A}^k_{1}\\
\vdots&\ddots& \vdots\\
{\psi_{A}}^1_{k}& \cdots &{\psi_A}^k_k
\end{array}
\right).& 
\end{align} 

The supersymmetric transformation $\delta$ of our model is given as follows.
\begin{align}
&\delta z^i_s=\psi^i_s,& &\delta \psi^i_s= i\phi_{\bar i m}z^m_s,&&\delta \chi^i_s =H^i_s,&&\delta H^i_s=(\gamma\delta_{\bar i j}+i\phi_{\bar i j}) \chi^j_s,&\no\\ 
&\delta A^i_j=\psi^i_{A j}, &
&\delta \phi_{\bar i j}=0,&&\delta \bar \phi_{\bar i j} =\eta_{\bar i j},&&\delta \psi_A=\gamma A+[A,i\phi].
\label{gsusy}
\end{align}
In the above transformation, $\gamma$ is central charge of central extension of standard supersymmetric transformation. We assume that $\gamma$ is a non-zero constant. 
 Supersymmetric transformation for complex conjugate of $X$ is defined as complex conjugate of $\delta X$ ($\delta \overline{X}:=\overline{(\delta X)}$). $\delta$ behaves like a grassmann variable. 
By applying (\ref{gsusy}) to (\ref{gmq}), we can obtain explicit form of the Lagrangian (\ref{gmq2}).
Note that the $z$ variables satisfy (\ref{stiefel}), the defining equations of  the Stiefel manifold $V_k(\mathbb{C}^N)$. On the other hand, $\psi$, the super-partner of $z$, plays the role of the $1$-form $d z$ in the supersymmetric path-integral. Since $ (\bm{z}^{ i})^\dag \bm{z}^j - \delta^{\bar i j}=0$ ($i,j=1,2,\cdots,k$), we also have the constraint for $dz$:
\begin{eqnarray}
\sum_{s=1}^N (dz^{\bar i}_{\bar s}z^j_s+z^{\bar i}_{\bar s}dz^j_s)=0\;\;\;(i,j=1,2,\cdots,k).
\end{eqnarray}
 By identifying $dz^i_s$ and  $dz^{\bar i}_{\bar s}$ with  $\psi^i_s$ and $\psi^{\bar i}_{\bar s}$. respectively, 
the above constraint is realized by insertion of the following projection operator $\omega$:
 \begin{align}
\omega&=\prod_{i,j=1}^k \Bigl[\sum_{s=1}^N(\psi^{\bar i}_{\bar s} z^j_s+z^{\bar i}_{\bar s} \psi^j_s) \Bigr].
\end{align}
$\bt$ is the normalization factor that normalizes volume of $G(k,N)$ into $1$.    
\begin{align}
\bt&=\frac{\prod_{j=0}^{k-1} j!}{2^{2k}  (-\pi) ^{k^2+kN}\pi^{kN+\frac{k(k+1)}{2}}(-1)^{\frac{k}{2}(k-1)}}.
\end{align}
Note that the volume $U(k)$ is given by, 
\begin{align}
\mathrm{vol}(U(k))&=\prod_{j=1}^k \mathrm{vol}(S^{2j-1})=\frac{2^k\pi^\frac{k(k+1)}{2}}{\prod_{j=1}^{k-1}j!}. 
\end{align}
Then, we can rewrite $\bt$ into the following form: 
\begin{align}
\bt&=\frac{1}{2^{k}  (-\pi) ^{k^2+kN}\pi^{kN}(-1)^{\frac{k}{2}(k-1)}\mathrm{vol}(U(k))}.
\end{align}
Lastly, we remark that our Lagrangian is not strictly  invariant under the supersymmetric transformation, i.e., $\delta \mathcal{L}_{MQ} \neq 0$, because our supersymmetric transformation is not nilpotent.
For example, we have
the following relation.
\begin{align*}
\delta^2<\psi_A,A>&=\delta \{ \mathrm{tr}(A^\dag(\gamma A+[i\phi,A])-\psi^\dag_A\psi_A ) \} \\
&=\mathrm{tr}(\psi_A^\dag(\gamma A+[i\phi,A])+A^\dag(\gamma\psi_A+[i\phi,\psi_A])
-(\gamma A+[A,i\phi])^\dag \psi_A+\psi_A^\dag (\gamma A+[A,i\phi]) )\\
&=2\gamma<\psi_A ,A>.
\end{align*}
This comes from central extension of the standard supersymmetry. The reason why we introduce central extension will be given in the next subsection.   
In some sense, we consider central extension of supersymmetric transformation in order to obtain top Chern class of tangent bundle of $G(k,N)$ from the toy model version of topological Yang-Mills theory.

\subsection{The New Feature of Our Model}
Explicit evaluation of the above path-integral will be given in Section 3 and Section 4, but 
in this subsection, we briefly explain new feature of our model which has not appeared in the former literatures on MQ 
formalism.  Let us introduce well-known facts on Chern classes of  holomorphic tangent bundle $T^{\prime}G(k,N)$ of $G(k,N)$. 
 Let $S$ be  tautological bundle of  $G(k,N)$ whose fiber of $\Lambda\in G(k,N)$ is given by complex  $k$-dimensional subspace $\Lambda\subset \mathbb{C}^N$ itself ($\mathrm{rk} (S)=k$). Then universal quotient bundle $Q$ ( $\mathrm{rk}Q=N-k$) is defined by the following exact sequence
\begin{align}
0\to S \to \mathbb{C}^N  \to Q \to 0. \label{exact 1}
\end{align}
where $\mathbb{C}^N$ means trivial bundle $G(k,N)\times \mathbb{C}^N$. Since $T^{\prime}G(k,N)$ can be identified with $Q \otimes S^*$, 
we obtain the following exact sequence:
\begin{align}
0\to S\otimes S^* \to \mathbb{C}^N \otimes S^* \to T^{\prime}G(k,N) \to 0.
\end{align}
Hence total Chern class of $T^{\prime}G(k,N)$ is given by $\displaystyle{\frac{(c(S^*))^N}{c(S\otimes S^*)}}$ (Euler number of $G(k, N)$ is obtained from integration of top Chern class of $T^{\prime}G(k,N)$).  If we decompose $S^*$ formally by the line bundle $L_i$ ($i=1,2,\cdots,k$):
\ba
S^{*}=\mathop{\oplus}_{i=1}^{k}L_{i},
\ea
$c(S^*)$ is written as $\prod_{i=1}^k(1+t x_i),\;\;(x_i:= c_1(L_i))$.  From splitting principle of Chern classes, we also have $c(S\otimes S^*)=\prod_{l>j}(1-t^2(x_l-x_j)^2)$. 
Then we obtain,
\begin{align}
c(T^{\prime}G(k,N))=\frac{\prod_{i=1}^k(1+t x_i)^N}{\prod_{l>j}(1-t^2(x_l-x_j)^2)}.
\label{topch}
\end{align}
Top chern class is given as coefficent of $t^{k(N-k)}$. 
The new feature of our model is given as follows.\\ 
\\
{\bf  By introducing central extension of standard supersymmetry with central charge $\gamma(\neq 0)$, we can produce top Chern class of $T^{\prime}G(k,N)$
via the total Chern class (\ref{topch}).}\\
Precisely speaking, integration of $H$ and $\chi$ results in  $\displaystyle{\prod_{i=1}^{k}(\gamma+x_{i})^N}$ and integration of $A$ and $\psi_{A}$ produces $\displaystyle{\frac{1}{\gamma^{k}\prod_{l>j}(\gamma^2-(x_l-x_j)^2)}}$. But after integration of the Grassmann variable $\psi$, which corresponds to 
integration of differential form on $G(k,N)$, only the contribution from the top Chern class survives and we obtain the Euler number of 
$G(k,N)$. As we will see later, final result does not depend on $\gamma$ as long as $\gamma$ is non-zero. 
This construction was obtained after many try and errors, but we don't know whether there exists more natural construction 
that evaluates Euler number of $G(k,N)$ by using MQ formalism and Atiyah-Jeffrey construction.     

\subsection{Organization of the paper}
This paper is organized as follows.
 In Section 2, we give an overview of MQ formalism and  Atiyah-Jeffrey construction. In Section 3,  we construct Lagrangian that counts Euler number of $G(k,N)$ by applying these techniques. Then we  integrate out fields except for $\psi$ and show that the Euler number is represented by fermion integral of the Chern class represented by the matrix $i\phi$ whose $(i,j)$-element is given by $<\bm{\psi}^i,\bm{\psi}^j>$.  In Section 4, we prove our main theorem by showing that the representation of the Chern class by the fermion variables give the desired Chern class as an elements of  cohomology ring of $G(k,N)$.  Especially, validity of normalization factor $\beta$ will  be verified. We think that combinatorial aspects in the discussions in Section 4 is quite interesting for mathematicians.  

\vspace{1cm}

\noindent
{\bf  Acknowledgements}    We would like to thank Prof. M.~Yoshinaga and Prof. Y.~Goto for valuable discussions. Research of M.J. is partially supported by JSPS KAKENHI 
Grant No. JP17K05214.

\section{The Mathai Quillen formalism and the Atiyah Jeffrey construction}
In this section, we explain outline of Mathai-Quillen formalism and Atiyah Jeffrey construction. For more details, see the   
literatures \cite{Bla1,Lab1,Sako,WZ}.
\subsection{Overview of Mathai Quillen formalism}
 Mathai-Quillen formalism (MQ formalism) provides us with a recipe to construct Thom from of a vector bundle. Here, we briefly explain outline of Mathai-Quillen formalism. Let $\pi:E \to M$ be a vector bundle of rank $n$ on $n$-dimensional compact manifold $M$. 
We assume that each fiber $\pi^{-1}(x)$ has metric (or inner product) that varies smoothly as $x\in M$ varies. 
We denote by $\{f_1,\cdots,f_n\}$ a local orthonormal frame of $\pi^{-1}(U)$ ($U$ is some open subset of $M$) with respect to this metric. 
Let $\Omega^q (M,E)$ be vector space of $E$-valued differntial $q$-form on $M$ ($\Omega^0 (M,E) \simeq \Gamma(E)$ is vector space of smooth section of $E$) and $\triangledown^E: \Omega^q (M,E) \to \Omega^{q+1} (M,E)$ be a connection compatible with inner product on $E$. $\triangledown^E$ satisfies Leibniz rule for
 $g \in \Omega^q (M), s \in \Gamma(E)$
 \begin{align}
 \triangledown^E (gs) =(d_Mg)s+(-1)^q g\wedge (\triangledown^E s),
 \end{align}
where $d_M$ is exterier derivative on $M$ and $\Omega^q (M)$ is vector space of smooth $q$-form on $M$.
Let us define a connection form $\omega^j_i$ by, 
\begin{align}
\triangledown^E f_i=\omega_i^j f_j.
\end{align}
where $\omega_i^j=-\omega_j^i$. Curvature of $\triangledown^E$ is given by $(\triangledown^E)^2:=R^E$. $R^E$ for the local orthonormal frame is represented in the following form: 
\begin{align}
(R^E)_i^j f_j &:= (\triangledown^E)^2f_i=\triangledown^E(\omega_i^j f_j)=d_M\omega_i^j f_j -\omega_i^j \wedge \omega_j^kf_k \notag \\
&=(d_M\omega_i^j -\omega_i^k \wedge \omega_k^j)f_j. 
\end{align}

From now on, we assume $n=2m$. Let ${\bf u}=~^t(u^1,\cdots,u^{2m})$ be coordinates of fiber $\pi^{-1}(x)$ of $E$ and  
$\chi$ be  Grassmann variable: $\chi=~^t(\chi^1,\cdots,\chi^{2m})$, which correspond to super-partner of ${\bf u}$. And let $R_{ij}$ be $(R^{E})^j_i$ ($R_{i j}$ is skew-symmetric). Note that $R_{ij}$ is locally a 2-form on $M$.
With this set-up, Thom form $\Phi_\nabla(E)$ constructed in MQ-formalism is given as follows. 
\begin{align}
\Phi_\nabla(E)=\frac{1}{(2\pi)^m} e^{-|\bm{u}|^2/2} \int D\chi \exp \left( \frac{1}{2} ~^t \chi R \chi + i~^t(\nabla u)\chi \right). \\
|\bm{u}|^2:=\sum_{i=1}^{2m} (u^i)^2 , 
~^t(\nabla u)\chi:=\sum_{i=1}^{2m}(\nabla u)^i\chi^i:=\sum_{i=1}^{2m}(d_Eu^{i}+\omega^{ i}_j u^j)\chi^i,
\end{align}
where $d_E=d_M+ \sum_{i=1}^{2m}du^i \frac{\d}{\d u^i}$ is exterior derivative on $E$. 
Let $\mathcal{L}_0$ be a Lagrangian defined by,  
\begin{align}
\mathcal{L}_0:=|\bm{u}|^2/2-\frac{1}{2} ~^t \chi R \chi - i~^t(\nabla u)\chi. 
\end{align}
Then we define supersymmetric transformation as follows:
\begin{align}
\delta \chi^i := i u^i,\,\delta u^i := \nabla u^i.
\end{align}
Here, we assume the following.
 \begin{itembox}[l]{Assumptions for $\delta$}
 1. $\delta$ behaves like fermionic variable. Hence $\delta$ is anti-commutative with $d_E$.
 
 2. $\delta$ acts only fiber variables and $\delta \omega=\delta R=0$.
 \end{itembox}
\vspace{1em}
Then we can show that $\mathcal{L}_0$ invariant under $\delta$ transformation.  
\begin{align}
\delta(|\bm{u}|^2/2)&=\sum_{i=1}^{2m} u^i (\nabla u)^i. \\
\delta(-\frac{1}{2} ~^t \chi R \chi) &=\frac{1}{2} ( -\delta(\chi^i) R_{ i j} \chi^j+\chi^i R_{i j} \delta(\chi^j)) =- i u^i R_{i j} \chi^j. \\
\delta(- i~^t(\nabla u)\chi)&= - i\sum_{i=1}^{2m}(\delta((\nabla u)^i)\chi^i- (\nabla u)^i\delta(\chi^i)) \\
&=- i\sum_{i=1}^{2m}((-d_E\delta(u^i)-\omega^{ i}_j\delta(u^j))\chi^i- i(\nabla u)^iu^i) \\
&=- i\sum_{i=1}^{2m}((-d_E(\omega^{i}_j u^j)-\omega^{ i}_j(d_E u^j+\omega^{ j}_k u^k))\chi^i- i(\nabla u)^iu^i)\\
&=- i\sum_{i=1}^{2m}(-(d_E(\omega^{ i}_j) u^j -\omega^{ k}_j\omega^{ i}_k u^j))\chi^i- i(\nabla u)^iu^i) \\
&\text{~~~(Since $\omega^{ i}_j$ is locally a $1$-form on $M$, $d_E(\omega^{ i}_j) =d_M\omega^{ i}_j$.)}\\
&=\sum_{i=1}^{2m}(iR_{ j i}u^j\chi^i-(\nabla u)^iu^i)
=iR_{ j i}u^j\chi^i-\sum_{i=1}^{2m}(\nabla u)^iu^i.
\end{align}
Hence $\delta \mathcal{L}_0 =0$. 
Let us  integrate out $\Phi_\nabla(E)$ on a fiber $\pi^{-1}(x)$. Since $x\in M$ is fixed, $R=\omega=0$.  
Then we can derive,
\begin{align}
\int_{\pi^{-1}(x)} \Phi_\nabla(E)=1.
\label{t1}
\end{align}
Explicit derivation is given as follows.
\begin{align}
\int_{\pi^{-1}(x)} \Phi_{\nabla}(E) &=(2\pi)^{-m} (-1)^{m+\frac{(2m-1)(2m)}{2}} \int_{\pi^{-1}(x)} e^{-|\bm{u}|^2/2} \int \mathcal{D} \chi  \prod_{a=1}^{2m} (1+idu^a \chi^a)  \notag \\
&= \frac{(-1)^{\frac{(2m-1)(2m)}{2}}}{(2\pi)^m } \int_{\pi^{-1}(x)} e^{-|\bm{u}|^2/2} \int \mathcal{D} \chi  (du^1 \chi^1) \cdots (du^{2m} \chi^{2m}) \notag \\
&= \frac{1}{(2\pi)^m }   \int_{\pi^{-1}(x)} e^{-|\bm{u}|^2/2}  du^1  \cdots du^{2m}=1.
\end{align}
(\ref{t1}) is one of the two features that characterizes Thom form of $E$. The other one is given by,
\ba
s_{0}^{*}(\Phi_\nabla(E))=e_{0,\nabla}(E),
\ea 
where $s_{0}:M\rightarrow E$ is the zero section of $E$ and $e_{0,\nabla}(E)$ is Euler class of $E$. This can be easily seen as follows:
\ba
s_{0}^{*}(\Phi_\nabla(E))&=&\frac{1}{(2\pi)^{m}}\int D\chi \exp \left( \frac{1}{2} {}^t \chi R \chi \right)\no\\
                               &=&\frac{1}{(2\pi)^{m}} \mbox{Pfaff}(R)\no\\
                               &=&e_{0,\nabla}(E),
\ea
where we used that $u^{i}(s_{0}(x))=0$ and that $R$ is skew symmetric. Integration of Euler class $e_{0,\nabla}(E)$ on $M$ gives Euler 
number of $E$, which is denoted by $\chi(E)$. Therefore, we have,
\ba
\chi(E)=\int_{M}s_{0}^{*}(\Phi_\nabla(E)).
\ea

At this stage, we include auxiliary bosonic variable $H^i$ 
and modify the supersymmetric transformation as follows.
\ba
\delta \chi^i :=H^i,\;\;\delta H^i :=R_{ij} \chi^j.
\label{susy}
\ea
Now we intoroduce $ \Psi :=\left< \chi, \frac{H}{2}-i u \right>$ where $<A,B>:=~^tAB$ is inner product of $A$ and $B$. We also use  
the notation $|A|^2:=<A,A>$.
Then $\delta \Psi$ is given as follows.
\begin{align}
\delta \Psi & =\delta \left< \chi, \frac{H}{2}-i u \right>=~^t(\delta \chi) \left(\frac{H}{2}-i u \right)-~^t\chi \left(\frac{\delta H}{2}-i (\delta u) \right) \notag \\
&=\frac{1}{2}\sum_{i=1}^{2m}(H^i-i u^i)^2 +\frac{|\bm{u}|^2}{2}-\frac{1}{2} ~^t \chi R \chi - i~^t(\nabla u)\chi
=\frac{1}{2}|H-i u|^2+\mathcal{L}_0.
\end{align}
$\delta \Psi$ can be identified with $\mathcal{L}_{0}$ modulo the relation $H^a=i u^a$ (equation of motion of $H$).
We can easily see that $\Phi_{\nabla}(E)$ is obtained by integrating $\exp(-\delta \Psi)$ by $H$ and $\chi$.
\begin{align}
\Phi_{\nabla}(E) &:= \frac{1}{(2\pi)^{2m}} \int \mathcal{D}\chi \int \mathcal{D}H \exp\qty( - \frac{1}{2} | H- iu |^2 -\mathcal{L}_0 ) \notag \\
&=\frac{1}{(2\pi)^{2m}} \int \mathcal{D}\chi \int \mathcal{D}H \exp\qty( - \delta \left\lag \chi, \frac{H}{2} -iu  \right\rag ).
\end{align}
Let $s:M\rightarrow E$ is any smooth section of $E$, then we can easily derive, 
\begin{align}
e_{s, \nabla}(E) := s^* \Phi_{\nabla}(E) = \frac{1}{(2\pi)^{2m}} \int \mathcal{D}\chi \mathcal{D}H \exp\qty( - \delta \left\lag \chi, \frac{H}{2} -is  \right\rag ). \label{eq:e_s,conn} 
\end{align}
Since $s$ is homotopic to the zero section $s_{0}$, $e_{s,\nabla}$ belongs to the same cohomology class as $e_{0,\nabla}(E)$. 
Hence we obtain the following equality.
\ba
\chi(E)&=& \int_{M}e_{0,\nabla}(E)\no\\
         &=& \int_{M}e_{s,\nabla}(E)\no\\
&=&\frac{1}{(2\pi)^{2m}} \int_{M} \mathcal{D}x \mathcal{D}\psi \mathcal{D}\chi \mathcal{D}H \exp\left( - \delta \left\lag \chi, \frac{H}{2} -is  \right\rag \right). \label{eq:MQ_original_form}
\ea
where $x$ is local coordinate of $M$ and $\psi$ is fermion variable that plays the role of differential form $dx$ on $M$.

\subsection{The Case when $M$ is an Orbit Space}
In this subsection, we search for Lagrangian that produces Euler class of $M$ when $M$ is given as an orbit space  
, by using  Atiyah-Jeffrey construction. Atiyah-Jeffrey construction is an extension of MQ formalism for 
vector bundle whose base space $M$ is given as an orbit space $X/G$ ($G$: Lie group) . 
Originally, Atiyah and Jeffrey constructed their formalism to study mathematical background to Witten's construction of 
Lagrangian of Topological Yang-Mills theory \cite{AJ1,Wit1}. In \cite{AJ1}, the orbit space is given by
$\mathcal{A}/\mathcal{G}$ where $\mathcal{A}$ is an infinite dimensional space of $SU(2)$ connections ($A_{\alpha}$) on a compact 4-dimensional manifold $M_{4}$ and $\mathcal{G}$ is gauge transformation group that acts on $\mathcal{A}$. 
The vector bundle $\pi: \mathcal{E}\rightarrow \mathcal{A}/\mathcal{G}$ is not clearly stated in \cite{AJ1}, but  
the section $s:\mathcal{A}/\mathcal{G} \to \mathcal{E}$ is given by, 
\ba
s(A_{\alpha})&=&F_{\alpha\beta}+\frac{1}{2}\epsilon_{\alpha\beta\gamma\delta}F^{\gamma\delta},\no\\
              &&(F_{\alpha\beta}=\d_{\alpha}A_{\beta}-\d_{\beta}A_{\alpha}+[A_{\alpha},A_{\beta}]).
\ea
Zero locus of the above section is nothing but the moduli space of anti self-dual instantons on $M_{4}$, and it connects 
topological Yang-Mills theory with Donaldson invariants. Then how was the supersymmetric transformation (\ref{susy}) 
modified to fit into the situation of an orbit space. Let us consider first behavior of $\delta^{2}$ of the transformation 
(\ref{susy}). 
\ba
\delta^2 \chi^i :=R_{ij}\chi^j,\;\;\delta^2 H^i :=R_{ij} H^j.
\label{sr}
\ea
Therefore, $\delta^{2}$ corresponds to ``infinitessimal rotation of fiber coordinates generated by $R_{ij}$''. 
Then Atiyah and Jeffrey modified the above relation into the following:   
\begin{align}
\delta^2=\delta_\phi,
\end{align}
where $\delta_\phi$ is the infinitesimal gauge transformation generated by $\phi$. 
It corresponds to infinitesimal rotation of infinite dimensional Lie group $\mathcal{G}$.   
Note that $\delta$ is nilpotent when we consider orbit space $\mathcal{A}/\mathcal{G}$. Then $\delta$ can be regarded as 
 infinite dimensional version of equivariant derivative $d-\iota_\omega,\,\omega \in \mathrm{Lie}(\mathcal{G})$ on $X/G$.
\begin{align}
\delta \Leftrightarrow d-\iota_\omega, \\
\delta^2=\delta_\phi \Leftrightarrow (d-\iota_\omega)^2=-d \iota_\omega-\iota_\omega d=-\mathcal{L}_\omega,
\end{align}
where $\mathcal{L}_\omega$ is the Lie derivative. With these considerations, $\delta$ transformation is modified as follows.
\begin{align}
&\delta x=\psi,& &\delta \psi= \delta_\phi x,&&\delta \chi =H,&&\delta H=\delta_\phi \chi,
&\delta \phi=0,&
\end{align}
where $x$ is coordinate of $\mathcal{A}$, $\psi$ is the fermion coordinate that plays the role of $dx$, $\chi$ is fermion coordinate of $\mathcal{E}$ and $H$ is auxiliary field and super-partner of $\chi$. 

In the previous subsection, $\mathcal{L}_{MQ}$, the Lagrangian obtained from the MQ-formalism was represented as  $\delta \Psi :=\delta <\chi,\frac{H}{2}-is>$. In Atiyah-Jeffrey construction, Lagrangian $\mathcal{L}_{MQ}$, which is expected to produce Euler number $\chi( \mathcal{E})$ of the vector bundle $\mathcal{E}$ on $\mathcal{A}/\mathcal{G}$,
\ba
\chi(\mathcal{E})\stackrel{?}{=}\int\mathcal{D}x\mathcal{D}\psi\mathcal{D}H\mathcal{D}\chi\mathcal{D}\phi\exp(-\mathcal{L}_{MQ}),
\label{ajmq}
\ea 
is given by,
\ba
\mathcal{L}_{MQ}=\delta( \Psi+\Psi_{\mathrm{proj}}).
\ea
The term $\delta \Psi_{\mathrm{proj}}$ plays the role of projecting out gauge horizontal direction (direction parallel to orbit of 
$\mathcal{G}$) in integrating $\psi$ over $T^{*}\mathcal{A}$. In other words, $\exp(-\delta \Psi_{\mathrm{proj}})$ can 
be regarded as projection operator from $T^{*}\mathcal{A}$ to $T^{*}(\mathcal{A}/\mathcal{G})$. 
In (\ref{ajmq}), we add $?$ above $=$ because $\mathcal{E}$, the vector bundle of infinite rank, is not clearly stated in \cite{AJ1} and 
$\chi(\mathcal{E})$ is not well-defined. Indeed, the Lagrangian $\mathcal{L}_{MQ}$ is used to produce Donaldson invariants of $M_{4}$ 
in context of topological Yang-Mills theory \cite{Wit1}. 

Let us explain outline of construction of $\Psi_{\mathrm{proj}}$. On $\mathcal{G}$-principle bundle $\mathcal{A} \to \mathcal{A}/\mathcal{G}$, group action for $x \in \mathcal {A}$ is given by $\mathcal{G}$. Let  $C:\mathfrak{g}\; (\mbox{Lie algebra of } \mathcal{G} ) \to T_x \mathcal{A}$ be differentiation of group action on $x \in \mathcal{A}$. Let $\theta$ be 
an element of $\mathfrak{g}$. Then, $C\theta$ is given by,
\begin{align}
C\theta= \delta_{\theta} x.
\end{align}
$C^\dag$ is defined as adjoint operator of C,
\begin{align}
\left< C^\dag \psi,\theta \right>=\left<\psi,C \theta \right>,
\end{align}
where $\left<*,*\right>$ in the l.h.s. is inner product of $\mathfrak{g}$ and  $\left<*,*\right>$ in the r.h.s is inner product of $T^{+}\mathcal{A}$.
If $\psi \in \mathrm{Ker} C^\dag$, we obtain
\begin{align}
C^\dag \psi=0 \Leftrightarrow 
0=\left< C^\dag \psi,\theta \right>=\left<\psi,C \theta \right>=\left<  \psi,\delta_\theta x \right>.
\end{align}
So, $\mathrm{Ker}C^\dag \subset T^{*}\mathcal{A}$ corresponds to vertical direction of gauge transformation.
Then we have to restrict integration of $\psi$ into $\mathrm{Ker}C^\dag$.
At this stage, we introduce additional boson field $\bar \phi$ and fermion field $\eta$.
Supersymmetric transformation for these fields are defined by,
\begin{align}
&\delta \bar \phi =\eta,&& \delta \eta= \delta_ \phi \bar \phi,&
\end{align}
where $\phi$ is gauge transformation parameter (element of $\mathfrak{g}$). Then $\Psi_{\mathrm{proj}}$ is defined in the following form.
\begin{align}
\Psi_{\mathrm{proj}} := \left< \psi, C\bar \phi \right>.
\end{align}
We obtain
\begin{align}
\delta \Psi_{\mathrm{proj}} =\delta \left<C^\dag \psi, \bar \phi \right>
=\left<\delta(C^\dag \psi), \bar \phi \right>-\left<C^\dag \psi, \delta(\bar \phi) \right>
=\left<\delta(C^\dag \psi), \bar \phi \right>-\left<C^\dag \psi, \eta \right>.
\end{align}
From equation of motion of $\eta$, we obtain
 \begin{align}
 \frac{\delta}{\delta \eta} \delta \Psi_{\mathrm{proj}}=0 \Leftrightarrow C^\dag \psi=0.
 \end{align}
Hence multiplying $\exp\left (-\delta \Psi_{\mathrm{proj}} \right)$ can restrict $\psi$ integration to $\mathrm{Ker}C^\dag$.
When we consider the case of zero section ($s\equiv 0$), the Lagrangian becomes, 
\begin{align}
\mathcal{L}_{MQ}=\delta \Bigl<\chi,\frac{H}{2}\Bigr>+i \delta\left<\psi,C\bar \phi\right>,
\end{align}
where supersymmetric transformation is given by,
\begin{align}
&\delta x=\psi,& &\delta \psi= \delta_\phi x,&&\delta \chi =H,&&\delta H=\delta_\phi \chi,&  \notag \\
&\delta \phi=0,&&\delta \bar \phi =\eta,&&\delta \eta=\delta_\phi \bar \phi.
\end{align}
This is our starting point of construction of the Lagrangian (\ref{gmq}) and of the supersymmetric transformation (\ref{gsusy}).
But what we aim to compute is $\chi(G(k,N))=\chi(T^{\prime}G(k,N))$, we have to modify the above settings further. 
Some points of modification are already mentioned in  Subsection 1.1 and Subsection 1.2. We will discuss again points of 
modification in Subsection 3.2.

\section{Construction of Lagrangian and First Half of Evaluation of Path-Integral}
\subsection{The Grassmann Manifold }
 In order to apply MQ formalism on the Grassmann manifold, it is necessary to consider the Grassmann manifold $G(k,N)$ as an orbit space $X/G$ \cite{Wit2}. $G(k,N)$ is a space which parametrizes all $k$-dimensional linear subspaces of the $N$-dimensional complex vector space $\mathbb{C}^N$. 
\begin{align}
G(k,N):= \{ W \subset \mathbb{C}^N | \dim_{\mathbb{C}} W=k\}.
\end{align}  
Then we introduce the Stiefel manifold $V_k(\mathbb{C}^N)$. A point of $V_k(\mathbb{C}^N)$ is given by a set of $k$ unit vectors 
in $\mathbb{C}^N$ which are orthogonal to each other.
 \begin{align}
 V_k(\mathbb{C}^N)&:=\{ (\bm{z}^1,\cdots,\bm{z}^k) \in \mathbb{C}^{kN}|\; \bm{z}^i \in \mathbb{C}^N,\,<\bm{z}^i, \bm{z}^j>=\delta_{i j} \,\,(i,j=1,\cdots,k)\} 
 \end{align}
Let $<\bm{z}^1,\cdots, \bm{z}^k>_{\mathbb{C}}$ be a vector space spanned by $\bm{z}^1, \cdots, \bm{z}^k$. 
Two points $(\bm{z}^1, \cdots ,\bm{z}^k)$ and $(\bm{z}^{\prime 1}, \cdots, \bm{z}^{\prime k})$ in $V_k(\mathbb{C}^N)$ satisfy the relation $<\bm{z}^1,\cdots, \bm{z}^k>_{\mathbb{C}}= <\bm{z}^{\prime 1},\cdots, \bm{z}^{\prime k}>_{\mathbb{C}}$
if and only if there exists $(U_{i j})_{1\leq i,j \leq k} \in U(k)$
that satisfy $\bm{z}^i=\sum_{j=1}^k U_{i j} \bm{z}^{\prime j}\,\, (i=1,\cdots,k)$.
Hence we can identify $G(k,N)$ with the orbit space $V_k(\mathbb{C}^N)/U(k)$.

Since $V_k(\mathbb{C}^N)$ is regraded as quotient space $U(N)/U(N-k)$ and volume of $U(N)$ is given by 
$\displaystyle{\prod_{j=1}^{N}\mathrm{vol}(S^{2j-1})}$ ($S^{2j-1}$ is the $(2j-1)$-dimensional unit sphere) \cite{Cam1, Fuj1}, we obtain,    
\begin{align}
\mathrm{vol}(V_k(\mathbb{C}^N))=\prod_{j=N-k+1}^{N}\mathrm{vol}(S^{2j-1})=\frac{2^k (\pi )^{kN-\frac{k(k-1)}{2}}}{\displaystyle \prod_{j=N-k+1}^{N} (j-1)!}.
\label{volst}
\end{align}

\subsection{Lagrangian that counts Euler Number $\chi(G(k,N))$} \label{LEN} 
In this section, we construct the Lagrangian for $\mathcal{L}_{MQ}$ by applying MQ-formalism and Atiyah-Jeffrey construction
outlined in the previous section to the orbit space $G(k,N)=V_k(\mathbb{C}^N)/U(k)$ From now on, we set section $s$ used 
in Subsection 2.2 to zero section.

Let us mention again the fields used in our Lagrangian.
The variable $z$ that describes a point of $V_k(\mathbb{C}^N)$ is given as follows.
\begin{align}
&z:=(\bm{z}^1\cdots \bm{z}^k) \in  \mathbb{C}^{kN},\\
&\bm{z}^j:=~^t(z^j_1,z^j_2,\cdots,z^j_N),\,\,(j=1,2,\cdots,k),\,\,
 (\bm{z}^{ i})^\dag \bm{z}^j - \delta^{\bar i j}=0 \,(i,j=1,2,\cdots,k).\label{eq1}
\end{align}
This corresponds to the variable $x$ in Subsection 2.2.
 Other fields are represented in the following form.
\begin{align}
&\psi:=(\bm{\psi}^1 \cdots \bm{\psi}^k),& &\bm{\psi}^j:=~^t(\psi^j_1,\psi^j_2,\cdots,\psi^j_N).&
\\&\phi :=\left(\begin{array}{ccc}
\phi_{\bar 1 1}&\cdots &\phi_{\bar 1 k}\\
\vdots&\ddots& \vdots \\
\phi_{\bar k 1}& \cdots & \phi_{\bar k k}
\end{array} \right).&
&\bar \phi :=\left(\begin{array}{ccc}
\bar \phi_{\bar 1 1}&\cdots &\bar \phi_{\bar 1 k}\\
\vdots&\ddots& \vdots \\
\bar \phi_{\bar k 1}& \cdots &\bar \phi_{\bar k k}
\end{array} \right),& \\
&\eta:=\left(\begin{array}{ccc}
\eta_{\bar 1 1}&\cdots &\eta_{\bar 1 k}\\
\vdots&\ddots& \vdots \\
\eta_{\bar k 1}& \cdots &\eta_{\bar k k}
\end{array} \right).& \\
&\chi:=\left(\begin{array}{ccc}
\bm{\chi}^1& \cdots & \bm{\chi}^k
\end{array}
\right), &
&\bm{\chi}^j:=~^t(\chi^j_1,\cdots,\chi^j_N).&\\
&H:=\left(\begin{array}{ccc}
\bm{H}^1& \cdots & \bm{H}^k
\end{array}
\right), &
&\bm{H}^j:=~^t(H^j_1,\cdots,H^j_N).& 
\end{align} 
$\phi$ and $\bar \phi$ are Hermite matrix $(\phi_{\bar i j}=\phi_{j \bar i}, \bar \phi_{\bar i j}=\bar \phi_{j \bar i})$. They are the generators of the elements of $U(k)$. $\bar \phi$ is not the complex conjugate of $\phi$. These fields play the same roles as the 
corresponding fields in Subsection 2.2.
Next, we introduce the field $A$ and its superpartner $\psi_A$ in order to produce the part $\displaystyle{\frac{1}{c(S\otimes S^{*})}}$ in the total Chern class $\displaystyle{\frac{c(S^{*})^{N}}{c(S\otimes S^{*})}}$.
\begin{align}
&A:=
\left(\begin{array}{ccc}
A_{\bar 1 1}&\cdots &A_{\bar 1 k}\\
\vdots&\ddots& \vdots\\
A_{\bar k1}& \cdots &A_{\bar k k}
\end{array}
\right),& 
&\psi_A:=
\left(\begin{array}{ccc}
\psi^1_{A1}&\cdots &\psi^k_{A1}\\
\vdots&\ddots& \vdots\\
\psi^1_{Ak}& \cdots &\psi^k_{Ak}
\end{array}
\right).& 
\end{align}
At this stage, we define supersymmetric transformation for each variables.  Supersymmetry are represent in following forms.
\begin{align}
&\delta z^i_s=\psi^i_s,& &\delta \psi^i_s= i\phi_{\bar i m}z^m_s,&&\delta \chi^i_s =H^i_s,&&\delta H^i_s=(\gamma\delta_{\bar i j}+i\phi_{\bar i j}) \chi^j_s,&\no\\ 
&\delta A^i_j=\psi^i_{A j}, &
&\delta \phi_{\bar i j}=0,&&\delta \bar \phi_{\bar i j} =\eta_{\bar i j},&&\delta \psi_A=\gamma A+[A,i\phi].
\label{susy1}
\end{align}
These are fundamentally obtained from applying the construction in Subsection 2.2 to the orbit space $G(k,N)=V_k(\mathbb{C}^N)/U(k)$
but we intoroduce central charge $\gamma(\neq 0)$ and modify the transformation of $H$ and $\psi_{A}$ from the standard version $\delta H^i_s=i\phi_{\bar i j} \chi^j_s$ and 
$\delta \psi_{A}=[A,i\phi]$. This modification corrsponds to central extension of supersymmetry algebra.

With these set-up's Lagrangian is defined as follows:
\begin{align}
\mathcal{L}_{MQ}=\delta <\chi,H >+\delta<\psi_A,A>+\frac{i}{2}\delta\{<\psi,C\bar \phi>+~^*<\psi,C\bar \phi>\}.
\end{align}
where $C\bar \phi=\delta_{\bar \phi} z=iz{}^{t}\bar \phi$. Except for the term $\delta<\psi_A,A>$, this Lagrangian is 
obtained from applying discussion of Subsection 2.2  with $s=0$ to the orbit space $G(k,N)=V_k(\mathbb{C}^N)/U(k)$. 
Then we can derive (\ref{gmq2}) from (\ref{gmq}) by straightforward computation.
Lastly, we write again the path-integral that will be evaluated in the remaining part of this paper.
\begin{align*}
Z_{MQ}:=&\bt\int_{V_k(\mathbb{C}^N)} Dz \int D\psi D\phi D\bar \phi DA D\eta D\chi DH   \,\,\omega \exp(-\mathcal{L}_{MQ}) \\
\mathcal{L}_{MQ}&=\delta<\chi,H>+\delta<\psi_A,A>+\frac{i}{2}\delta\Bigl[<\psi,iz\bar \phi>+\overline{<\psi,iz\bar \phi>}\Bigr] 
\label{gmq3}\\
&=\sum^N_{s=1}\Bigl[\sum^k_{i=1}H^{\bar i}_{\bar s} H^i_s-\bigl(\chi^{\bar i}_{\bar s}(\gamma\delta_{\bar i j}+i\phi_{\bar i j} )\chi^j_s \bigr) +\psi^{\bar i}_{\bar s} \bar \phi _{\bar i l} \psi^{ l}_s \Bigr]+i\mathrm{tr}(\phi \bar \phi) \notag \\
&+\frac{1}{2}\sum_{s=1}^N \left\{ \psi^{\bar i}_{\bar s} z^{ j}_s-\psi^{ j}_{ s}  z^{\bar  i}_{\bar s}\right\}\eta_{\bar i j} +\mathrm{tr}(A^\dag(\gamma A+[i\phi,A])-\psi^\dag_A\psi_A ).
\end{align*}
As we have already mentioned,  $\bt$ is normalization factor for Chern class. and $\omega$ is projective operator along tangent space of $V_k(\mathbb{C}^N)$. They are given as follows. 
\begin{align}
\omega&=\prod_{i,j=1}^k \Bigl[\sum_{s=1}^N(\psi^{\bar i}_s z^j_s+z^{\bar i}_s \psi^j_s) \Bigr], ~~~
\bt=\frac{\prod_{j=0}^{k-1} j!}{2^{2k}  (-\pi) ^{k^2+kN}\pi^{kN+\frac{k(k+1)}{2}}(-1)^{\frac{k}{2}(k-1)}}.
\end{align}

\subsection{First Half of Evaluation of the Path-Integral} \label{CEN}

\subsubsection{$U(N)\times U(k)$ Symmetry of the Lagrangian}

\begin{lemma}
\label{transz}
\begin{align}
Z_{MQ}&=\bt\int_{V_k(\mathbb{C}^N)} Dz \int D\psi D\phi D\bar \phi DA D\eta D\chi DH   \,\,\omega \exp(-\mathcal{L}_{MQ})\no \\
&=\bt \mathrm{vol}(V_k(\mathbb{C}^N))\int  D\psi D\phi D\bar \phi DA D\eta D\chi DH   \,\,\omega^{\prime} 
\exp(-\mathcal{L}_{MQ}^{\prime}). \no\\
\mathcal{L}_{MQ}^{\prime}&=\sum^N_{s=1}\Bigl[\sum^k_{i=1}H^{\bar i}_{\bar s} H^i_s-\bigl(\chi^{\bar i}_{\bar s}(\gamma\delta_{\bar i j}+i\phi_{\bar i j} )\chi^j_s \bigr) +\psi^{\bar i}_{\bar s} \bar \phi _{\bar i l} \psi^{ l}_s \Bigr]+i\mathrm{tr}(\phi \bar \phi) \notag \\
&+\frac{1}{2}\sum_{i,j=1}^k \left\{ \psi^{\bar i}_{\bar j} -\psi^{ j}_{ i}  \right\}\eta_{\bar i j} +\mathrm{tr}(A^\dag(\gamma A+[i\phi,A])-\psi^\dag_A\psi_A ). \no\\
\omega^{\prime}&=\prod_{i=1}^k \prod_{j=1}^k \Bigl[(\psi^{\bar i}_{\bar j} + \psi^j_i) \Bigr].
\end{align}
\end{lemma}

\begin{proof}
 First, we note that projection operator $\omega$ is rewritten as follows: 
 \begin{align}
&\omega=\prod_{i=1}^k \prod_{j=1}^k (\sum_{s=1}^N (\psi^{\bar i}_{\bar s} z^j_s+\psi^{j}_{ s} z^{\bar i}_{\bar s}))\no\\
&=(-1)^{\frac{k^2}{2}(k^2+1)} \int D\theta \exp \left ( <\psi,z~^t\theta>+ \overline{<\psi,z~^t\theta>}\right ),
\end{align}
where $\theta$ is fermionic  Hermite matrix.

Hence we can rewrite $Z_{MQ}$ by, 
\begin{eqnarray}
&&Z_{MQ}=(-1)^{\frac{k^2}{2}(k^2+1)}\bt\int_{V_k(\mathbb{C}^N)} Dz \int D\psi D\phi D\bar \phi DA D\eta D\chi DH D\theta  \;\;\exp(-\widetilde{\mathcal{L}_{MQ}}(z,\psi,\phi,\bar\phi,A,\eta,\chi,H,\theta)) \no\\
&&\widetilde{\mathcal{L}_{MQ}}(z,\psi,\phi,\bar\phi,A,\eta,\chi,H,\theta)
=\mathcal{L}_{MQ}(z,\psi,\phi,\bar\phi,A,\eta,\chi,H)+<\psi,z~^t\theta>+ \overline{<\psi,z~^t\theta>}.
\end{eqnarray}
If we transform $\chi, \psi, \phi, \bar \phi, \eta, A, z, H,\theta$ in the following way,
\begin{align}
&\chi=U^{N}\chi^\prime U^{k}, &&\psi=U^{N} \psi^\prime U^{k},&&\phi=~^tU^k \phi^\prime ~^*U^k,&&\bar \phi=~^tU^k \bar \phi^\prime ~^*U^k,
\no\\
&\eta=~^tU^k \eta^\prime ~^*U^k,&&A=~^tU^k A^\prime ~^*U^k,&&z=U^N z^{\prime} U^k,&&H=U^N H^{\prime} U^k,\no\\
&\theta=~^tU^k \theta^\prime ~^*U^k,
\end{align}
we can easily see that $\widetilde{\mathcal{L}_{MQ}}$ has $U(N)\times U(k)$ symmetry, i.e., 
\ba
&&\widetilde{\mathcal{L}_{MQ}}(z^{\prime},\psi^{\prime},\phi^{\prime},\bar\phi^{\prime},A^{\prime},\eta^{\prime},\chi^{\prime},H^{\prime},
\theta^{\prime})\no\\
&&=\widetilde{\mathcal{L}_{MQ}}(z,\psi,\phi,\bar\phi,A,\eta,\chi,H,\theta).
\label{unuk}
\ea
On the other hand, we  can confirm that the integral measures of each variable is also $U(N)\times U(k)$ invariant: $DX=DX^\prime$ ($X$ is each variable). 
Let us define $F_{MQ}(z)$ by,
\ba
F_{MQ}(z):=(-1)^{\frac{k^2}{2}(k^2+1)}\bt\int D\psi D\phi D\bar \phi DA D\eta D\chi DH D\theta  \;\;\exp(-\widetilde{\mathcal{L}_{MQ}}(z,\psi,\phi,\bar\phi,A,\eta,\chi,H,\theta)).
\ea
Then we obviously have,
\ba
Z_{MQ}=\int_{V_k(\mathbb{C}^N)}Dz F_{MQ}(z).
\ea 
Then by using (\ref{unuk}) and invariance of integral measure, we obtain,
\ba
&&F_{MQ}(z)\no\\
&&=(-1)^{\frac{k^2}{2}(k^2+1)}\bt\int D\psi D\phi D\bar \phi DA D\eta D\chi DH D\theta  \;\;\exp(-\widetilde{\mathcal{L}_{MQ}}(z,\psi,\phi,\bar\phi,A,\eta,\chi,H,\theta))\no\\
&&=(-1)^{\frac{k^2}{2}(k^2+1)}\bt\int D\psi D\phi D\bar \phi DA D\eta D\chi DH D\theta\exp(-\widetilde{\mathcal{L}_{MQ}}(z^{\prime},\psi^{\prime},\phi^{\prime},\bar\phi^{\prime},A^{\prime},\eta^{\prime},\chi^{\prime},H^{\prime},
\theta^{\prime}))\no\\
&&=(-1)^{\frac{k^2}{2}(k^2+1)}\bt\int D\psi^{\prime} D\phi^{\prime} D\bar \phi^{\prime} DA^{\prime} D\eta^{\prime} D\chi^{\prime} 
DH^{\prime} D\theta^{\prime}\;\;\exp(-\widetilde{\mathcal{L}_{MQ}}(z^{\prime},\psi^{\prime},\phi^{\prime},\bar\phi^{\prime},A^{\prime},\eta^{\prime},\chi^{\prime},H^{\prime},
\theta^{\prime}))\no\\
&&=F_{MQ}(z^{\prime})
\ea
For each $z \in V_k(\mathbb{C}^N)$, we can choose  $U^N \in U(N)$ and $U^k \in U(k)$ that satisfy, 
\begin{align}
U^N z U^k=\left (\begin{array}{c}
I_k\\0_{N-k,k} 
\end{array}
\right)=:z_{0},
\end{align}
where $I_k$ is  $k\times k$-type unit matrix and $0_{N-k,k}$ is $(N-k)\times k$-type zero matrix. 
Hence we obtain,
\ba
Z_{M,Q}=\int_{V_k(\mathbb{C}^N)}Dz F_{MQ}(z)=\int_{V_k(\mathbb{C}^N)}Dz F_{MQ}(z_{0})=\mathrm{vol}(V_k(\mathbb{C}^N))F_{MQ}(z_{0}).
\ea
This is nothing but the assertion of the Lemma.
\end{proof}

\input Eulernumber

\section{Second Half: Proof of the Main Theorem}
\subsection{Free Fermion Realization of Cohomology Ring of $G(k,N)$}
In the previous section, we reached the expression (\ref{euler}). 
Then what remains to prove is the following equality:
\ba
Z_{MQ}&=&\frac{ \prod_{j=0}^{k-1} j!}{ \prod_{j=N-k}^{N-1} j!}\int D\psi  \frac{\left(\det(\gamma I_k+\Phi)\right)^N }{\gamma^{k} \prod_{l>j} \left(\gamma^2-(\lambda^\prime_l-\lambda^\prime_j)^2 \right)}  \prod_{i=1}^k \prod_{j=1}^k \psi^{ j}_{ i} \psi^{\bar j}_{\bar i}\\ 
&=&\int_{G(k,N)}c(T^{\prime}G(k,N))\\
&=&\int_{G(k,N)}\frac{\prod_{i=1}^{k}(1+x_{i})}{\prod_{l>j}(1-(x_{l}-x_{j})^2)}\\
&=&{N\choose k},
\ea
where $c(S^{*})=\prod_{i=1}^{k}(1+x_{i})$.
First, we note that the factor  $\prod_{i=1}^k \psi^j_i \psi^{\bar j}_{\bar i}$ allows us to neglect $\psi^j_i, \;\psi^{\bar j}_{\bar i}
\;(i,j=1,2,\cdots,k)$ in the integral measure and the remaining part of the integrand. 
Hence we only have to consider the fields $\psi^j_i, \;\psi^{\bar j}_{\bar i}\;\;(i=k+1,\cdots,N, j=1,\cdots,k)$.
At this stage, we redefine $\psi^j_{k+i}, \;\psi^{\bar j}_{\bar k+\bar i}$ by  $\psi^j_{i}, \;\psi^{\bar j}_{\bar i}$
$(i=1,\cdots,N-k,\;j=1,\cdots k)$ and introduce
\begin{align}
\Phi^\prime&:=\sum_{s=1}^{N-k} \left( \begin{array}{ccc}
\omega_s^{1 \bar1} &\ldots&\omega^{1\bar k }_s\\
\vdots& \ddots&\vdots\\
\omega_s^{k \bar 1}&\ldots&\omega_s^{k \bar k}\\
\end{array}
\right)\qquad (\omega_{s}^{i\bar j}:=\psi_{s}^{i}\psi_{\bar s}^{\bar j}), \\
D\psi^\prime&=\prod_{s=1}^{N-k} d\psi^1_s d\psi^{\bar 1}_{\bar s} \cdots d\psi^k_s d\psi^{\bar k}_{\bar s}.
\end{align}
Let $\lambda_{i}\; (i=1,\cdots,k)$ be eigenvalues of the Hermite matrix $\Phi^{\prime}$.
Then (\ref{euler}) is rewritten as follows.
\begin{align}
Z_{MQ}
&=\frac{ \prod_{j=0}^{k-1} j!}{ \prod_{j=N-k}^{N-1} j!}\int D\psi^\prime \frac{\left(\det(\gamma I_k+\Phi^\prime)\right)^N }{\gamma^{k}\prod_{l>j} \left(\gamma^2-(\lambda_l-\lambda_j)^2 \right)}.
\label{zfinal}
\end{align}

\begin{theorem} \label{Th1}
Let us define $b_{i}\;(i=0,1,2,\cdots)$ by, 
\begin{align}
\frac{1}{\mathrm{det}(I_k+t\Phi^\prime)}=\sum_{m=0}^\infty b_m t^m.
\end{align}
Then, $b_m=0$ if $m>N-k$.
\end{theorem}
\begin{proof}
By using Gaussian integral of complex variables $X_{1},\cdots, X_{k}$, we obtain the equality:
\begin{align}
\frac{1}{\det(I_k+t\Phi^\prime)}&=\int_{\mathbb{C}^n} \mathcal{D} X \exp{(- ~^t\bar X (I_k+t\Phi^\prime)X)}. 
\end{align}
where $X={}^{t}(X_1,\cdots,X_{k})$ and integral measure is given by 
$\mathcal{D}X:=\frac{1}{(-2\pi i)^k}dX_1d X_{\bar 1}\cdots dX_k dX_{\bar k}$.
\begin{align}
&(R.H.S)=\int_{\mathbb{C}^n} \mathcal{D} X \exp{(- X_{\bar i}(\delta_{\bar i j}+t\sum_{\mu=1}^{N-k} \omega^{ i \bar j}_\mu)X_j)} \notag \no\\ 
&=\int_{\mathbb{C}^n} \mathcal{D} X \{e^{-|X|^2} e^{-t\sum_{\mu=1}^{N-k} \omega^{i \bar j}_\mu X_{\bar i} X_j } \}\no\\
&=\int_{\mathbb{C}^n} \mathcal{D} X e^{-|X|^2} \sum_{m=0}^\infty \frac{1}{m!} \Bigl( -t\sum_{\mu=1}^{N-k} \omega^{i\bar j}_\mu X_{\bar i} X_j \Bigr)^m \no\\
&=\sum_{m=0}^\infty \frac{(-t)^m}{m!}  \sum_{(\mu_1,\cdots,\mu_m)} \sum_{\substack{(i_1,\cdots,i_m)\\ (j_1,\cdots,j_m)}}\omega^{ i_1 \bar j_1}_{\mu_1}\cdots \omega^{ i_m \bar j_m}_{\mu_m} \int_{\mathbb{C}^n} \mathcal{D} X e^{-|X|^2}  X_{\bar i_1} X_{j_1}\cdots X_{\bar i_m} X_{j_m}.
\label{gauss}
\end{align}
Here, $|X|^2$ represents $\sum_{i=1}^k X_{\bar i} X_i$.
We remark that $1\leq \mu_{l}\leq N-k,\;\;1\leq i_{l},j_{l}\leq k,\;(l=1,\cdots,m)$. 
By Wick's theorem of Gaussian integral, we can easily see that the integral in the last line of (\ref{gauss}) does not vanish  
if and only if $\{i_1,\cdots,i_m\}=\{j_1,\cdots,j_m\}$.
Hence we obtain, 
\begin{align}
&\int_{\mathbb{C}^n} \mathcal{D} X \exp{(- ~^t\bar X (I_k+t\Phi^\prime)X)}\no\\
&=\sum_{m=0}^\infty \frac{(-t)^m}{m!}  \sum_{(\mu_1,\cdots,\mu_m)} \sum_{(i_{1},\cdots,i_m)} \sum_{\sigma \in S_{m}} 
\mbox{Sym}(i_{1},\cdots,i_{m})\omega^{ i_1 \bar i_{\sigma(1)}}_{\mu_1}\cdots \omega^{ i_m \bar i_{\sigma(m)}}_{\mu_m}\int_{\mathbb{C}^n} \mathcal{D} X e^{-|X|^2} \prod_{j=1}^{m} |X_{i_j}|^{2}. 
\label{key}
\end{align}
where $\mbox{Sym}(i_{1},\cdots,i_{m})$ is symmetric factor of the m-tuple $(i_1,\cdots,i_{m})$ given by,
\ba
\mbox{Sym}(i_{1},\cdots,i_{m})&=&\prod_{j=1}^{k}\frac{1}{\mbox{mul}((i_{1},\cdots,i_{m});j)!},\no\\
\mbox{mul}((i_{1},\cdots,i_{m});j)&:=&(\mbox{ number of $l$'s that satisfy $i_{l}=j$}).
\ea
Then let us fix $(\mu_{1},\cdots,\mu_{m})$ and $(i_{1},\cdots, i_{m})$, and  assume that there exists a pair $(i,j)\;(1\leq i<j\leq m)$ that satisfy $\mu_{i}=\mu_{j}$. Without loss of generality, we can further assume that $\mu_{1}=\mu_{2}=\mu$.
Obviously,  for any $\sigma \in S_m$  we can uniquely take $\sigma^\prime \in S_m$ that satisfy $\sigma^{\prime}(1)=\sigma(2),\;\sigma^{\prime}(2)=\sigma(1),\;\sigma^{\prime}(i)=\sigma(i) (i=3,4,\cdots m)$. 
Then, we can easily see, 
\ba
&&\omega^{i_1 \bar i_{\sigma(1)}}_{\mu}\omega^{i_2 \bar i_{\sigma(2)}}_{\mu}\omega^{i_3 \bar i_{\sigma(3)}}_{\mu_{3}} 
\cdots \omega^{ i_{m} \bar i_{\sigma(m)}}_{\mu_m}
+\omega^{i_1\bar i_{\sigma^{\prime}(1)}}_{\mu}\omega^{i_2 \bar i_{\sigma^{\prime}(2)}}_{\mu}\omega^{i_3 \bar i_{\sigma^{\prime}(3)}}_{\mu_{3}}\cdots \omega^{ i_{m} \bar i_{\sigma^{\prime}(m)}}_{\mu_m}\no\\
&=&\omega^{i_1 \bar i_{\sigma(1)}}_{\mu}\omega^{i_2 \bar i_{\sigma(2)}}_{\mu}\omega^{i_3 \bar i_{\sigma(3)}}_{\mu_{3}} 
\cdots \omega^{ i_{m} \bar i_{\sigma(m)}}_{\mu_m}
+\omega^{i_1\bar i_{\sigma(2)}}_{\mu}\omega^{i_2 \bar i_{\sigma(1)}}_{\mu}\omega^{i_3 \bar i_{\sigma(3)}}_{\mu_{3}}\cdots \omega^{ i_{m} \bar i_{\sigma(m)}}_{\mu_m}\no\\
&=&0
\ea
because $\omega^{ i_1\bar  i_{\sigma(2)}}_{\mu}\omega^{ i_2\bar  i_{\sigma(1)}}_{\mu}
= \psi_\mu^{ i_1} \psi_{\bar\mu}^{\bar i_{\sigma(2)}} \psi_{\mu}^{i_2} \psi_{\bar\mu}^{\bar i_{\sigma(1)}}
=-\psi_\mu^{ i_1} \psi_{\bar\mu}^{\bar i_{\sigma(1)}} \psi_\mu^{i_2} \psi_{\bar\mu}^{\bar i_{\sigma(2)}}=-\omega^{ i_1\bar  i_{\sigma(1)}}_{\mu}\omega^{ i_2\bar  i_{\sigma(2)}}_{\mu}$.
Hence $$\sum_{\sigma \in S_{m}} 
\mbox{Sym}(i_{1},\cdots,i_{m})\omega^{ i_1 \bar i_{\sigma(1)}}_{\mu_1}\cdots \omega^{ i_m \bar i_{\sigma(m)}}_{\mu_m}$$
in the last line of (\ref{key}) vanishes if some $\mu_{j}$'s in the m-tuple $(\mu_{1},\cdots,\mu_{m})$ coincide.
Since $1\leq \mu_{j}\leq N-k$, it follows that the summand in the last line of (\ref{key}) vanishes if $m> N-k$. 
\end{proof} 
Then we introduce another theorem which will be proved in the next subsection. 
\begin{theorem}\label{ThV}
\begin{align}
\frac{ \prod_{j=0}^{k-1} j!}{ \prod_{j=N-k}^{N-1} j!}\int D\psi^\prime  \left(\det(\Phi^\prime)\right)^{N-k} =1.
\end{align}
\end{theorem}
Now we recall the well-known facts on the cohomology ring of $G(k,N)$ \cite{Bot1}. 
As we have mentioned in Subsection 1.2, we have the following exact sequence of vector bundles on $G(k,N)$.
\begin{align}
0\to S \to \mathbb{C}^N \to Q \to 0, \label{ges}
\end{align}
where $S$ is the tautological bundle of rank $k$, $\mathbb{C}^N$ is trivial bundle $G(k,N)\times \mathbb{C}^N$ and 
$Q$ is the quotient bundle of rank $N-k$.
Then the cohomology ring of $G(k,N)$ is given by, 
\begin{align}
H^{*}(G(k,N))=\frac{\R[c_1(S),\cdots,c_k(S),c_1(Q),\cdots,c_{N-k}(Q)]}{(c(S)c(Q)=1)}.
\end{align}
Since relation between $i$-th Chern class of  a vector bundle $E$ and the one  of its dual bundle $E^*$ is 
given by, $c_i(E^*)=(-1)^ic_i(E)$, we can take $c_{i}(S^{*})$'s  and $c_{i}(Q^{*})$'s as generators of  $H^{*}(G(k,N))$.               
\begin{align}
H^{*}(G(k,N))=\frac{\R[c_1(S^*),\cdots,c_k(S^*),c_1(Q^*),\cdots,c_{N-k}(Q^*)]}{(c(S^*)c(Q^*)=1)}. 
\label{crgd}
\end{align}
On the other hand, the relation $c(S^*)c(Q^*)=1$ is rewritten by, 
\begin{align}
&c(Q^*)=\frac{1}{c(S^*)}.
\label{grrel}
\end{align}
If we expand $\displaystyle{\frac{1}{c(S^*)}=1/(1+c_{1}(S^{*})t+c_{2}(S^{*})t^2+\cdots+c_{k}(S^{*})t^k)}$ in powers of $t$,
\ba
\frac{1}{1+c_{1}(S^{*})t+c_{2}(S^{*})t^2+\cdots+c_{k}(S^{*})t^{k}}=\sum_{i=0}^{\infty}a_{i}t^{i},
\ea
we can rewrite (\ref{grrel}) as follows:
\ba
c_{i}(Q^{*})=a_{i}\;\;(i=1,2,\cdots,N-k),\;\;a_{i}=0\;\;(i>N-k).
\ea
Note that $a_{i}$ is degree $i$ homogeneous polynomial of $c_{j}(S^{*})$'s $(j=1,2,\cdots,k)$.  
Hence we can eliminate generators $c_{j}(Q^{*})$'s from (\ref{crgd}) and obtain another representation of $H^{*}(G(k,N))$.
\begin{align}
H^{*}(G(k,N))=\frac{\R[c_1(S^*),\cdots,c_k(S^*)]}{(a_{i}=0\;\;(i>N-k))}. 
\label{crgd2}
\end{align} 
At this stage, we look back at Theorem \ref{Th1}.  Let us define $\sigma_{j}$ $(j=1,2\cdots,k)$ by,  
\begin{align}
1+\sigma_1t+\cdots+\sigma_k t^k:=\det(I_k+t \Phi^\prime)=\prod_{j=1}^{k}(1+\lambda_{j}t).
\end{align} 
In other words, $\sigma_{j}$ is the $j$-th fundamental symmetric polynomial of $\lambda_{1},\cdots,\lambda_{k}$.
Note that $\sigma_{k}$ is identified with $\det(\Phi^{\prime})$. 
Then let us consider the ring $\R[\sigma_{1},\cdots,\sigma_{k}]$.
Since we have set, 
\begin{align}
\sum_{i=0}^{N-k}b_it^i=\frac{1}{\det(I_k+t \Phi^\prime)}=\frac{1}{1+\sigma_1t+\cdots+\sigma_k t^k}, \label{qdp}
\end{align} 
assertion of Theorem \ref{Th1}: $b_{i}=0\;\;(i>N-k)$ tells us that we have a ring homomorphism $f:H^{*}(Gr(k,N))\rightarrow \R[\sigma_{1},\cdots,\sigma_{k}]$ defined by,  
\ba
f(c_{j}(S^{*}))=\sigma_{j}\;\;(j=1,2,\cdots,k).
\ea
\begin{theorem}
The ring homomorphism $f:H^{*}(Gr(k,N))\rightarrow \R[\sigma_{1},\cdots,\sigma_{k}]$ is an isomorphism. 
\label{isom}
\end{theorem}
\begin{proof}
Since (\ref{crgd2}) and Theorem \ref{Th1} holds true, $f$ is surjective. Then it is enough for us to show that $f$ is injective.    
Let us assume that $\mbox{Ker}(f)\neq \{0\}$. Then we have $a\neq 0$ that satisfies $f(a)=0$. Since $f$ preserves degree (of cohomology ring),
we can assume that $a\in H^{2j}(Gr(k,N))$. By Poincar\'e duality theorem, we have $b\in H^{2k(N-k)-2j}(Gr(k,N)$ that satisfies 
$ab=(c_k(S^{*}))^{N-k}\in H^{2k(N-k)}(Gr(k,N))$. On the other hand,  $f(c_k(S^{*})^{N-k})=f(c_k(S^{*}))^{N-k}=(\det(\Phi^{\prime}))^{N-k}$ and Theorem \ref{ThV} tells us 
that $f(c_k(S^{*})^{N-k})\neq 0$. But from the assumption $f(a)=0$, $f(ab)=f(a)f(b)=f(c_k(S^{*})^{N-k})=0$ follows. This contradiction leads 
us to conclude that $f$ is injective. 
\end{proof}
In Subsection 1.2, we have introduced formal line bundle decomposition $S^{*}=\mathop{\oplus}_{i=1}^{k}L_{i}$ and the relation:
\ba
c(S^{*})=\prod_{j=1}^{k}(1+x_{j}t)\;\;(x_{j}=c_{1}(L_{j})).
\ea 
Then $f(x_{j})=\lambda_{j}$, i.e., $f$ identifies $x_{j}$ with $\lambda_{j}$.
Moreover, according to \cite{GH}, normalization condition of integration
on $Gr(k,N)$ is given by,
\ba 
\int_{G(k,N)}(c_{k}(S^{*}))^{N-k}=1.
\ea
Therefore, Theorem \ref{ThV}  and Theorem \ref{isom} lead us to the following equality:
\ba
\int_{Gr(k,N)}g(x_{1},\cdots,x_{k})=\frac{ \prod_{j=0}^{k-1} j!}{ \prod_{j=N-k}^{N-1} j!}\int D\psi^\prime g(\lambda_{1},\cdots,\lambda_{k}),
\label{intrel}
\ea
where $g(x_{1},\cdots,x_{k})$ is a symmetric polynomial of $x_{1},\cdots,x_{k}$ that represents an element of $H^{2k(N-k)}(G(k,N))$.
By combining (\ref{zfinal}) with (\ref{intrel}), we obtain,  
\ba
Z_{MQ}
&=&\frac{ \prod_{j=0}^{k-1} j!}{ \prod_{j=N-k}^{N-1} j!}\int D\psi^\prime \frac{\left(\det(\gamma I_k+\Phi^\prime)\right)^N }{ \gamma^{k} \prod_{l>j} \left(\gamma^{2}-(\lambda_l-\lambda_j)^2 \right)}\no\\
&=&\frac{ \prod_{j=0}^{k-1} j!}{ \prod_{j=N-k}^{N-1} j!}\int D\psi^\prime \gamma^{k(N-k)}\frac{\left(\det(I_k+\frac{1}{\gamma}\Phi^\prime)\right)^N }
{\prod_{l>j} \left(1-\frac{1}{\gamma^{2}}(\lambda_l-\lambda_j)^2 \right)}\no\\
&=&\frac{ \prod_{j=0}^{k-1} j!}{ \prod_{j=N-k}^{N-1} j!}\int D\psi^\prime \gamma^{k(N-k)}\frac{\prod_{i=1}^{k}(1+\frac{1}{\gamma}\lambda_{i})^N }
{\prod_{l>j} \left(1-\frac{1}{\gamma^{2}}(\lambda_l-\lambda_j)^2 \right)}\no\\
&=&\frac{ \prod_{j=0}^{k-1} j!}{ \prod_{j=N-k}^{N-1} j!}\int D\psi^\prime \frac{\prod_{i=1}^{k}(1+\lambda_{i})^N }
{\prod_{l>j} \left(1-(\lambda_l-\lambda_j)^2 \right)}\no\\
&=&\int_{Gr(k,N)}\frac{\prod_{i=1}^{k}(1+x_{i})^N }
{\prod_{l>j} \left(1-(x_l-x_j)^2 \right)}\no\\
&=&\int_{Gr(k,N)}c(T^{\prime}G(k,N))\no\\
&=&\int_{Gr(k,N)}c_{top}(T^{\prime}G(k,N))\no\\
&=&\chi(G(k,N))={N\choose k}.
\ea
This completes proof of the main theorem. $\Box$

\subsection{Proof of Theorem \ref{ThV} }
\begin{definition}
Let ${\cal M}_{l}$ be set of $k\times k$ matrix $M_{l}$:
\ba
M_l:=\left( \begin{array}{ccc}
m^{l}_{1,1} &\ldots&m^{l}_{1,k}\\
\vdots& \ddots&\vdots\\
m^{l}_{k,1}&\ldots&m^{l}_{k,k}\\
\end{array}
\right),
\ea
whose $(i,j)$-element $m_{i,j}^{l}$ is given by non-negative integer that satisfies
the following conditions:
\ba
\sum_{i=1}^{k}m_{i,j}^{l}=l,\;\;(j=1,\cdots,k),\;\;\;
\sum_{j=1}^{k}m_{i,j}^{l}=l\;\;(i=1,\cdots,k).
\ea
\end{definition}
\begin{definition}
Let $S_{k}$ be symmetric group of size $k$. For $\sigma\in S_{k}$, we define $k \times k$ matrix $R(\sigma)$:
\begin{align}
R(\sigma):= \left( \begin{array}{ccc}
\delta_{\sigma(1),1} &\ldots&\delta_{\sigma(1),k}\\
\vdots& \ddots&\vdots\\
\delta_{\sigma(k),1}&\ldots&\delta_{\sigma(k),k}\\
\end{array}
\right),
\end{align}
where $\delta_{i,j}$ is Kronecker's delta symbol.
\end{definition}
\begin{proposition}
We denote by $(n_{\sigma})_{\sigma\in S_{k}}$ a sequence of $k!$ non-negative integers labeled by $\sigma\in S_{k}$.
Let ${\cal N}_{l}$ be set of $(n_{\sigma})_{\sigma\in S_{k}}$'s that satisfy $\sum_{\sigma\in S_{k}}n_{\sigma}=l$.
Then $\varphi:{\cal N}_{l}\to {\cal M}_{l}$ defined by,
\ba
\varphi((n_{\sigma})_{\sigma\in S_{k}})=\sum_{\sigma\in S_{k}}n_{\sigma}R(\sigma)\in {\cal M}_{l},
\ea 
is a surjection.
\end{proposition}
\begin{proof}
In the $l=1$ case, assertion of the proposition is obvious because $\varphi:{\cal N}_{1}\to {\cal M}_{1}$ is a bijection.
Then we can prove the proposition by induction of $l$. 
\end{proof}
\begin{remark}
For general $k$ and $l$, $\varphi$ is not injective. For example, in the $k=l=3$ case, we have the following equalities:
\ba
\left(\begin{array}{ccc}1&1&1\\1&1&1\\1&1&1\end{array}\right)
&=&\left(\begin{array}{ccc}1&0&0\\0&1&0\\0&0&1\end{array}\right)
+\left(\begin{array}{ccc}0&1&0\\0&0&1\\1&0&0\end{array}\right)+\left(\begin{array}{ccc}0&0&1\\1&0&0\\0&1&0\end{array}\right)
\no\\
&=&\left(\begin{array}{ccc}0&0&1\\0&1&0\\1&0&0\end{array}\right)
+\left(\begin{array}{ccc}0&1&0\\1&0&0\\0&0&1\end{array}\right)+\left(\begin{array}{ccc}1&0&0\\0&0&1\\0&1&0\end{array}\right).
\ea
\end{remark}

\begin{definition}
Let $X$ be $k\times k$ matrix whose $(i,j)$-element is given by $x_{i,j}$. For $M_{l}\in {\cal M}_{l}$, we define integer $\mul(M_{l})$ by 
the following expansion:
\ba 
 |X|^l=\sum_{M_l\in {\cal M}_{l}} \mul(M_{l}) \prod_{a,b=1}^{k}  x_{a,b}^{m^{l}_{a,b}}.
\label{muldef}
\ea
\end{definition}
\begin{proposition}
$\mul(M_{l})$ is explicitly 
evaluated as follows.
\ba
\mul(M_{l})=\sum_{\varphi((n_{\sigma})_{\sigma\in S_{k}})=M_{l}}\frac{l! \prod_{\sigma \in S_k}\Bigl(\mathrm{sgn}(\sigma)\Bigr)^{n_\sigma}}{\prod_{\sigma \in S_k} n_{\sigma}!}  
\label{Ml}
\ea 
\end{proposition}
\begin{proof}
First, we explicitly expand $|X|^{l}$ by using definition of $|X|$:
\ba
 |X|^l&=&\Bigl( \sum_{\sigma \in S_k} \mathrm{sgn}(\sigma) \prod_{a=1}^k x_{a,\sigma(a)} \Bigr)^{l} \no\\
&=&\sum_{(n_{\sigma})_{\sigma\in S_{k}}\in {\cal N}_{l}}\frac{l!}{\prod_{\sigma \in S_{k}} n_{\sigma}!} \prod_{\sigma \in S_{k}}\Bigl((\mathrm{sgn}(\sigma))^{n_{\sigma}}  \prod_{a=1}^{k} x_{a,\sigma(a)}^{n_{\sigma}} \Bigr) \no\\
&=&\sum_{(n_{\sigma})_{\sigma\in S_{k}}\in {\cal N}_{l}}\frac{l!}{\prod_{\sigma \in S_{k}} n_{\sigma}!} \prod_{\sigma \in S_{k}}\Bigl((\mathrm{sgn}(\sigma))^{n_{\sigma}}  \prod_{a,b=1}^{k} x_{a,b}^{n_{\sigma}\delta_{\sigma(a),b}} \Bigr) \no\\
&=&\sum_{(n_{\sigma})_{\sigma\in S_{k}}\in {\cal N}_{l}}\frac{l!}{\prod_{\sigma \in S_{k}} n_{\sigma}!} \prod_{\sigma \in S_{k}}\Bigl((\mathrm{sgn}(\sigma))^{n_{\sigma}}\bigr)\prod_{a,b=1}^{k}x_{a,b}^{\sum_{\sigma\in S_{k}}n_{\sigma}\delta_{\sigma(a),b}}. 
\label{pexp}
\ea
$\sum_{\sigma\in S_{k}}n_{\sigma}\delta_{\sigma(a),b}$ is nothing but the $(a,b)$-element of $\sum_{\sigma\in S_{l}}n_{\sigma}R(\sigma)$
and $\sum_{\sigma\in S_{l}}n_{\sigma}R(\sigma)=\varphi((n_{\sigma})_{\sigma\in S_{k}})\in {\cal M}_{l}$.
Then assertion of proposition immediately follows from (\ref{pexp}). 
\end{proof}

\begin{lemma}
The following equality holds, 
\begin{align}
\mul(M_{l+1})=\sum_{\sigma \in S_k} \mathrm{sgn}(\sigma) \mul(M_{l+1}-R(\sigma)), \label{lemma2}
\end{align}
where we set $\mul(M_{l+1}-R(\sigma))=0$ if $M_{l+1}-R(\sigma)\notin {\cal M}_{l}$.
\end{lemma}
\begin{proof}
\ba
|X|^{l+1}&=&\sum_{M_{l+1}\in{\cal M}_{l+1}} \mul(M_{l+1}) \prod_{a,b=1}^{k}  x_{a,b}^{m_{a,b}^{l+1}}\no\\
&=&|X|^{l}\cdot|X|\no\\ 
&=&\Biggl[\sum_{M_{l}\in {\cal M}_{l}} \mul(M_{l}) \prod_{a,b=1}^{k}  x_{a,b}^{m^{l}_{a,b}} \Biggr]\times \Biggl[\sum_{\sigma \in S_k} \mathrm{sgn}(\sigma) \prod_{a=1}^k x_{a \sigma(a)}\Biggr]\no
\ea
\ba
&=& \sum_{M_{l}\in{\cal M}_{l}} \sum_{\sigma \in S_k} \mathrm{sgn}(\sigma) \mul(M_{l}) \prod_{a,b=1}^{k}  x_{a,b}^{m_{a,b}^{l}+\delta_{\sigma(a) b}}\no \\
&=&\sum_{M_{l+1}\in {\cal M}_{l+1}} \sum_{\sigma \in S_k} \mathrm{sgn}(\sigma) \mul(M_{l+1}-R(\sigma)) \prod_{a,b=1}^{k}  x_{a,b}^{m_{a,b}^{l+1}} .\no
\ea
\end{proof}

\begin{definition}
Let $\psi_{s}^{i}\;\;(s=1,\cdots,l,\;i=1,\cdots,k)$ be complex Grassmann variable and $\psi_{\bar s}^{\bar i}$ be its complex conjugate. 
We denote by $\Phi^{\prime}$ a $k\times k$ matrix whose $(i,j)$-element is given by $\sum_{s=1}^{l}\psi_{s}^{i}\psi_{\bar s}^{\bar j}$. 
Then we define $C(l,k)$ as follows.
\begin{align}
&C(l,k):=\int \Bigl(\prod_{s=1}^{l}\prod_{i=1}^{k}d\psi_{s}^{i}d\psi_{\bar s}^{\bar i}\Bigr) \det(\Phi^{\prime})^{l}.
\label{defc}
\end{align}

\end{definition}

For brevity, we introduce the following notations.
\ba
D\psi&:=&\prod_{s=1}^{l}\prod_{i=1}^{k}d\psi_{s}^{i}d\psi_{\bar s}^{\bar i},\no\\
\omega_{s}^{i\bar j}&:=&\psi_{s}^{i}\psi_{\bar s}^{\bar j}.
\ea
We note here the equalities:
\ba
(\omega_{s}^{i\bar j})^2=\omega_{s}^{i\bar j}\omega_{s}^{i\bar k}=\omega_{s}^{i\bar j}\omega_{s}^{k\bar j}=0,
\ea
which plays an important role in proof of the next lemma. 
\begin{lemma}
\begin{align}
C(l,k)=\sum_{M_l\in{\cal M}_{l}}\Biggl[\prod_{a,b=1}^k (m^{l}_{a,b}) ! \Biggr] \mul(M_{l})^2. 
\label{eq:con1}
\end{align}
\label{clklem}
\end{lemma}

\begin{proof}
By substituting $\Phi^{\prime}$ for $X$ in (\ref{muldef}), we rewrite (\ref{defc}) as follows. 
\begin{align}
&C(l,k)=\int D\psi (\mathrm{det}\Phi^\prime)^{l} 
=\sum_{M_l\in{\cal M}_{l}} \mul(M_{l}) \int D\psi  \prod_{a,b=1}^{k} \Bigl(  \sum_{s=1}^{l} \omega_s^{ a \bar b} \Big)^{m^{l}_{a,b}} \notag \\
&=\sum_{M_l\in{\cal M}_{l}} \mul(M_{l}) \int D\psi  \prod_{a,b=1}^{k} \biggl( (m_{a,b}^{l})!\bigl(\sum_{1\leq s^1_{a,b}<\cdots<s^{m^{l}_{a,b}}_{a,b}\leq l}  \omega_{s^1_{a,b}}^{ a\bar b}\cdots \omega_{s^{m^l_{a,b}}_{a,b}}^{ a \bar b}\bigr)  \biggr) \notag \\
&=\sum_{M_l\in{\cal M}_{l}} \mul(M_{l}) \Biggl[ \prod_{a,b=1}^{k} (m^{l}_{a,b})!\Biggr]\int D\psi\Bigl( \sum_{1\leq s^1_{a,b}<\cdots<s^{m^{l}_{a,b}}_{a,b}\leq l}   \prod_{a,b=1}^{k}  \omega_{s^1_{a,b}}^{ a\bar b}\cdots \omega_{s^{m^l_{a,b}}_{a,b}}^{ a \bar b}  \Big).
\label{clk}
\end{align}
In going from the first line to the second line, we used an equality:
\ba
\Bigl(  \sum_{s=1}^{l} \omega_s^{ a \bar b} \Big)^{m^{l}_{a,b}}=(m_{a,b}^{l})!\bigl(\sum_{1\leq s^1_{a,b}<\cdots<s^{m^{l}_{a,b}}_{a,b}\leq l}  \omega_{s^1_{a,b}}^{ a\bar b}\cdots \omega_{s^{m^l_{a,b}}_{a,b}}^{ a \bar b}\bigr), 
\ea
which follows from $(\omega_{s}^{a\bar b})^2=0$.
Let us define $S_{a,b}:=\{s^1_{a,b},\cdots,s^{m_{a,b}^{l}}_{a,b}\}\;(|S_{a,b}|=m^{l}_{a,b})$
associated with the sequence $1\leq s^1_{a,b}<\cdots<s^{m^{l}_{a,b}}_{a,b}\leq l$ in the last line of (\ref{clk}). 
Since  $\omega_{s}^{i\bar j}\omega_{s}^{i\bar k}=\omega_{s}^{i\bar j}\omega_{s}^{k\bar j}=0$, $S_{ab}$'s associated with non-vanishing $\prod_{a,b=1}^{k}  \omega_{s^1_{a,b}}^{ a\bar b}\cdots \omega_{s^{m^l_{a,b}}_{a,b}}^{ a \bar b}$ satisfy the following condition:
\ba
\coprod_{a=1}^k S_{a,b}&=&\{1,\cdots, l\}\;\;\;(b=1,\cdots,k),\no\\
\coprod_{b=1}^k S_{a,b}&=&\{1,\cdots, l\}\;\;\;(a=1,\cdots,k).
\label{scond}
\ea
We denote by $\{S_{a,b}\}$ set of $S_{a,b}$'s $(a,b=1,\cdots,k)$ that satisfy $|S_{a,b}|=m_{a,b}^{l}$ and the condition (\ref{scond}).
We also denote by $S(M_{l})$ set of $\{ S_{a,b}\}$'s associated with the matrix $\bigl(m_{a,b}^{l}\bigr)=M_{l}\in {\cal M}_{l}$. 
Then we can further rewrite $C(l,k)$ into the following form:
\ba
C(l,k)=\sum_{M_l\in{\cal M}_{l}} \mul(M_{l}) \Biggl[ \prod_{a,b=1}^{k} (m^{l}_{a,b})!\Biggr]\int D\psi\Bigl( \sum_{\{S_{a,b}\}\in S(M_{l})}\prod_{a,b=1}^{k}  \omega_{s^1_{a,b}}^{ a\bar b}\cdots \omega_{s^{m^l_{a,b}}_{a,b}}^{ a \bar b}  \Big).
\ea
We take a closer look at the sum  $\displaystyle{\sum_{\{S_{a,b}\}\in S(M_{l})}\prod_{a,b=1}^{k}\omega_{s^1_{a,b}}^{ a\bar b}\cdots \omega_{s^{m^l_{a,b}}_{a,b}}^{ a \bar b}}$. For a fixed element $\{S_{a,b}\}\in S(M_{l})$, we can construct a sequence $(\sigma_{1},\sigma_{2},\cdots,\sigma_{l})$ of permutations $\sigma_{s}\in S_{k}\;\;(s=1,\cdots, l)$. This is because for each $s\in\{1,2,\cdots,l\}$, 
we can fix unique permutation $\sigma_{s}\in S_{k}$ that satisfy $s\in S_{a,\sigma_{s}(a)}\;\;(a=1,\cdots,k)$ by using the condition (\ref{scond}). Obviously, the sequence $(\sigma_{1},\sigma_{2},\cdots,\sigma_{l})$ satisfy the following condition:  
\ba
\sum_{s=1}^{l}R(\sigma_{s})=M_{l}.
\label{seqcond}
\ea
Conversely, for a sequence $(\sigma_{1},\cdots,\sigma_{l})$ that satisfy (\ref{seqcond}), we can construct unique $\{S_{a,b}\}\in S(M_{l})$ 
that satisfies $s\in S_{a,\sigma_{s}(a)}\;\;(a=1,\cdots,k)$. Hence we have one to one correspondence between $\{S_{a,b}\}\in S(M_{l})$
and a sequence $(\sigma_{1},\cdots,\sigma_{l})$ that satisfy (\ref{seqcond}).
Since we can construct  $\frac{l!}{\prod_{\sigma\in S_{k}} n_\sigma!}$ different elements of $S(M_{l})$ from a fixed $(n_{\sigma})_{\sigma\in S_{k}}$ that satisfies $\varphi((n_{\sigma})_{\sigma\in S_{k}})=M_{l}$, the following equality holds. 
$$
\int D\psi \biggl(\sum_{\{S_{a,b}\}\in S(M_{l})}\prod_{a,b=1}^{k}\omega_{s^1_{a,b}}^{ a\bar b}\cdots \omega_{s^{m^l_{a,b}}_{a,b}}^{ a \bar b}\biggr)=\sum_{\varphi((n_{\sigma})_{\sigma\in S_{k}})=M_{l}} \frac{l!}{\prod_{\sigma\in S_{k}} n_\sigma!} \int D\psi  \prod_{s=1}^{l}\omega_{s}^{1 \overline{\sigma_s(1)}}\cdots \omega_{s}^{ k \overline{\sigma_s(k)}}  
$$
Then we obtain,
\ba
&&C(l,k)\no\\
&&=\sum_{M_l\in{\cal M}_{l}} \mul(M_{l}) 
\Biggl[ \prod_{a,b=1}^{k} (m^{l}_{a,b})! \Biggr]
\left[\sum_{\varphi((n_{\sigma})_{\sigma\in S_{k}})=M_{l}}  \frac{l!}{\prod_{\sigma\in S_{k}} n_\sigma!} \int D\psi  \prod_{s=1}^{l} \Bigl(   \omega_{s}^{1 \overline{\sigma_s(1)}}\cdots \omega_{s}^{ k \overline{\sigma_s(k)}}  \Big) \right]\no\\
&&=\sum_{M_l\in{\cal M}_{l}} \mul(M_{l}) \Biggl[ \prod_{a,b=1}^{k} (m^{l}_{a,b})!\Biggr]
\left[\sum_{\varphi((n_{\sigma})_{\sigma\in S_{k}})=M_{l}} \frac{l!}{\prod_{\sigma\in S_{k}} n_\sigma !}   \prod_{\sigma} \mathrm{sgn}(\sigma)^{n_\sigma} \right]\no\\
&&=\sum_{M_l\in{\cal M}_{l}} \mul(M_{l}) \Biggl[ \prod_{a,b=1}^{k} (m^{l}_{a,b})!\Biggr] \mul(M_{l}).
\ea
In going from the third line to the last line, we used (\ref{Ml}).
\end{proof}
\begin{lemma} ({\bf cf.  \cite{Fuji2}})
\begin{align}
\sum_{\sigma  \in S_k}\mathrm{sgn}(\sigma)\Bigl(\prod_{a=1}^k \frac{\partial}{\partial x_{a, \sigma(a)}} \Bigr) |X|^{l+1} 
=\frac{(k+l)!}{l!}|X|^{l}. \label{lemma3}
\end{align}
\label{lemma5}
\end{lemma}

\begin{proof}
We start from the following identity:
\ba
|X|^{l+1}&=&\int\prod_{a=1}^{l+1}\prod_{i=1}^{k}d\psi_{\bar a}^{\bar i}d\psi_{a}^{i}\exp(\sum_{a=1}^{l+1}\sum_{i,j=1}^{k}\psi_{\bar a}^{\bar i}x_{i,j}\psi_{a}^{j})\no\\
&=&\int D\psi\exp(\sum_{i,j=1}^{k}x_{i,j}\omega^{i,j}).
\ea
In going from the first line to the second line, we set $D\psi=\prod_{i=1}^{k}d\psi_{a}^{i}d\psi_{\bar a}^{\bar i},\;\;\omega^{i,j}=\sum_{a=1}^{l+1}\psi_{\bar a}^{\bar i}\psi_{a}^{j}$.
Then we obtain,
\ba
&&\sum_{\sigma  \in S_k}\mathrm{sgn}(\sigma)\biggl(\prod_{a=1}^k \frac{\partial}{\partial x_{a,\sigma(a)}} \biggr) |X|^{l+1}\no\\
&=&\sum_{\sigma  \in S_k}\mathrm{sgn}(\sigma)\biggl(\prod_{a=1}^k \frac{\partial}{\partial x_{a, \sigma(a)}} \biggr)
\int D\psi\exp(\sum_{i,j=1}^{k}x_{i,j}\omega^{i,j})\no\\
&=&\int D\psi\exp(\sum_{i,j=1}^{k}x_{i,j}\omega^{i,j})\biggl(\sum_{\sigma  \in S_k}    \mathrm{sgn}(\sigma) \prod_{a=1}^{k}\omega^{a,\sigma(a)}\biggr).
\ea
Since the integrand in the last line is invariant under coordinate change $\psi_{a}^{i}\rightarrow P^{i}_{j}\psi_{a}^{j},\;\;
\psi_{\bar a}^{\bar i}\rightarrow  (P^{-1})^{i}_{j}\psi_{\bar a}^{\bar j}$ ($P$ is arbitrary $k \times k$ invertible matrix), we can assume that$X$ is a diagonal matrix $(x_{i,j}=
\lambda_{i}\delta_{i,j})$ from the start \footnote{Precisely speaking, $X$ is only reduced to Jordan normal form, but this does not 
affect the remaining part of the proof.}. 
Hence we obtain,
\ba
&&\sum_{\sigma  \in S_k}\mathrm{sgn}(\sigma)\biggl(\prod_{a=1}^k \frac{\partial}{\partial x_{a,\sigma(a)}} \biggr) |X|^{l+1}\no\\
&=&\int D\psi\exp(\sum_{i=1}^{k}\lambda_{i}\omega^{i,i})\biggl(\sum_{\sigma  \in S_k}\mathrm{sgn}(\sigma)\prod_{j=1}^{k}\omega^{j,\sigma(j)}\biggr)\no\\
&=&\int D\psi\frac{(\sum_{i=1}^{k}\lambda_{i}\omega^{i,i})^{kl}}{(kl)!}\biggl(\sum_{\sigma  \in S_k}\mathrm{sgn}(\sigma)\prod_{j=1}^{k}\omega^{j,\sigma(j)}\biggr)
\no\\
&=&\int D\psi\biggl(\sum_{\substack{n_1+n_2+\cdots+n_{k}=kl\\n_{1},\cdots,n_{k}\geq 0}}\prod_{i=1}^{k}\frac{(\lambda_{i}\omega^{i,i})^{n_{i}}}{n_{i}!}\biggr)\biggl(\sum_{\sigma  \in S_k}\mathrm{sgn}(\sigma)\prod_{j=1}^{k}\omega^{j,\sigma(j)}\biggr)\no\\
&=&\int D\psi \biggl(\prod_{i=1}^{k}\frac{(\lambda_{i}\omega^{i,i})^{l}}{l!} \biggr)\biggl(\sum_{\sigma  \in S_k}\mathrm{sgn}(\sigma)\prod_{j=1}^{k}\omega^{j,\sigma(j)}\biggr).
\label{step1}
\ea
In going from the fourth line to the last line, we used the condition that $\psi_{*}^{j}$ and $\psi_{\bar *}^{\bar j}$ must appear exactly 
at $l+1$ times for each $j$ for non vanishing result. We can further proceed as follows.
\ba
&&\int D\psi \biggl(\prod_{i=1}^{k}\frac{(\lambda_{i}\omega^{i,i})^{l}}{l!}\biggr) \biggl(\sum_{\sigma  \in S_k}\mathrm{sgn}(\sigma)\prod_{j=1}^{k}\omega^{j,\sigma(j)}\biggr)\no\\
&=&\prod_{i=1}^{k}(\lambda_{i})^{l}\int D\psi \biggl(\prod_{i=1}^{k}\frac{(\omega^{i,i})^{l}}{l!} \biggr) \biggl(\sum_{\sigma  \in S_k}\mathrm{sgn}(\sigma)\prod_{j=1}^{k}\omega^{j,\sigma(j)}\biggr)\no\\
&=&|X|^{l}\int D\psi \biggl(\prod_{i=1}^{k}\frac{(\omega^{i,i})^{l}}{l!}\biggr)\biggl(\sum_{\sigma  \in S_k}\mathrm{sgn}(\sigma)\prod_{j=1}^{k}\omega^{j,\sigma(j)}\biggr).
\label{step2}
\ea
At this stage, we recall the following equalities:
\ba
&&(\omega^{i,i})^{l}=(\sum_{a=1}^{l+1}\psi_{\bar a}^{\bar i}\psi_{a}^{i})^{l}=\sum_{a=1}^{l+1}l!(\prod_{b\neq a}\psi_{\bar b}^{\bar i}\psi_{b}^{i})
,\no\\
&&\sum_{\sigma  \in S_k}\mathrm{sgn}(\sigma)\prod_{j=1}^{k}\omega^{j,\sigma(j)}
=\sum_{a_{1}=1}^{l+1}\cdots\sum_{a_{k}=1}^{l+1}\sum_{\sigma  \in S_k}\mathrm{sgn}(\sigma)
\prod_{j=1}^{k}\psi_{\bar a_{j}}^{\bar j}\psi_{a_{j}}^{\sigma(j)}.
\label{step3}
\ea
By combining (\ref{step1}), (\ref{step2}) and (\ref{step3}), we obtain,
\ba
&&\sum_{\sigma  \in S_k}\mathrm{sgn}(\sigma)\biggl(\prod_{a=1}^k \frac{\partial}{\partial x_{a,\sigma(a)}} \biggr) |X|^{l+1}\no\\
&=&|X|^{l}\biggl(\sum_{\sigma  \in S_k}\sum_{a_{1}=1}^{l+1}\cdots\sum_{a_{k}=1}^{l+1}\int D\psi \left(\prod_{i=1}^{k}
\bigl(\sum_{a=1}^{l+1}\prod_{b\neq a}\psi_{\bar b}^{\bar i}\psi_{b}^{i}\bigr)\right) \mathrm{sgn}(\sigma)
\prod_{j=1}^{k}\psi_{\bar a_{j}}^{\bar j}\psi_{a_{j}}^{\sigma(j)}\biggr)\no\\
&=&|X|^{l}\biggl(\sum_{\sigma  \in S_k}\sum_{a_{1}=1}^{l+1}\cdots\sum_{a_{k}=1}^{l+1}\int D\psi \left(\prod_{i=1}^{k}
\bigl(\sum_{a=1}^{l+1}\prod_{b\neq a}\psi_{\bar b}^{\bar i}\psi_{b}^{i}\bigr) \right) \mathrm{sgn}(\sigma)\mathrm{sgn}(\sigma^{-1})
\prod_{j=1}^{k}\psi_{\bar a_{j}}^{\bar j}\psi_{a_{\sigma^{-1}(j)}}^{j}\biggr)\no\\
&=&|X|^{l}\biggl(\sum_{\sigma  \in S_k}\sum_{a_{1}=1}^{l+1}\cdots\sum_{a_{k}=1}^{l+1}\int D\psi \left(\prod_{i=1}^{k}
\bigl(\sum_{a=1}^{l+1}\prod_{b\neq a}\psi_{\bar b}^{\bar i}\psi_{b}^{i}\bigr)\right)
\prod_{j=1}^{k}\psi_{\bar a_{j}}^{\bar j}\psi_{a_{\sigma^{-1}(j)}}^{j}\biggr)\no\\
&=&|X|^{l}\biggl(\sum_{\sigma  \in S_k}\sum_{a_{1}=1}^{l+1}\cdots\sum_{a_{k}=1}^{l+1}\int D\psi \left(\prod_{i=1}^{k}
\bigl(\sum_{a=1}^{l+1}\prod_{b\neq a}\psi_{\bar b}^{\bar i}\psi_{b}^{i}\bigr) \right)
\prod_{j=1}^{k}\psi_{\bar a_{j}}^{\bar j}\psi_{a_{\sigma(j)}}^{j}\biggr).
\label{step4}
\ea
In going from the second line to the third line, we used the relation,
\ba
\prod_{j=1}^{k}\psi_{\bar a_{j}}^{\bar j}\psi_{a_{j}}^{\sigma(j)}=\mathrm{sgn}(\sigma^{-1})\prod_{j=1}^{k}\psi_{\bar a_{j}}^{\bar j}\psi_{a_{\sigma^{-1}(j)}}^{j}.
\ea
At this stage, we fix $\sigma\in {S_{k}}$ and consider the integral:
\ba
\sum_{a_{1}=1}^{l+1}\cdots\sum_{a_{k}=1}^{l+1}\int D\psi \left(\prod_{i=1}^{k}
\bigl(\sum_{a=1}^{l+1}\prod_{b\neq a}\psi_{\bar b}^{\bar i}\psi_{b}^{i}\bigr)\right) 
\prod_{j=1}^{k}\psi_{\bar a_{j}}^{\bar j}\psi_{a_{\sigma(j)}}^{j}.
\ea
We can easily see that it counts number of sequences $(a_{1},a_{2},\cdots,a_{k})$ $(1\leq a_{j}\leq l+1)$ that satisfy the 
following condition:
\ba
a_{j}=a_{\sigma(j)}\;\;(j=1,2,\cdots,k).
\label{acond}
\ea 
Let $P_{k}$ be set of partition of $k$,
\ba
P_{k}:=\{\tau=(n_{1},\cdots,n_{l(\tau)})\;|\;n_1+n_2+\cdots+n_{l(\tau)}=k,\;\;n_{1}\geq n_{2}\geq \cdots\geq n_{l(\tau)}>0\;\}.
\ea 
It is well-known that conjugacy class of symmetric group $S_{k}$ is labeled by $\tau\in P_{k}$. If $\sigma\in S_{k}$ belongs to 
the conjugacy class labeled by $\tau=(n_{1},\cdots,n_{l(\tau)})$, number of sequences $(a_{1},a_{2},\cdots,a_{k})$ that satisfy 
(\ref{acond}) is given by $(l+1)^{l(\tau)}$. It is well-known that number of permutations that belong to conjugacy class labeled by $\tau=(n_{1},\cdots,n_{l(\tau)})$ is given by, 
\ba
\frac{k!}{\prod_{j=1}^{l(\tau)}n_{j}\prod_{j=1}^{k}\mathrm{mul}(j;\tau)!},
\ea 
where $\mathrm{mul}(j;\tau)$ is number of $i$'s that satisfy $n_{i}=j$. Hence we have obtain,
\ba
&&\sum_{\sigma  \in S_k}\sum_{a_{1}=1}^{l+1}\cdots\sum_{a_{k}=1}^{l+1}\int D\psi \left(\prod_{i=1}^{k}
\bigl(\sum_{a=1}^{l+1}\prod_{b\neq a}\psi_{\bar b}^{\bar i}\psi_{b}^{i}\bigr) \right)
\prod_{j=1}^{k}\psi_{\bar a_{j}}^{\bar j}\psi_{a_{\sigma(j)}}^{j}\no\\
&=&\sum_{\tau\in P_{k}}\frac{k!}{\prod_{j=1}^{l(\tau)}n_{j}\prod_{j=1}^{k}\mathrm{mul}(j;\tau)!}(l+1)^{l(\tau)}.
\label{step5}
\ea
On the other hand, we have the following combinatorial identity:
\ba
\sum_{\tau\in P_{k}}\frac{k!}{\prod_{j=1}^{l(\tau)}n_{j}\prod_{j=1}^{k}\mathrm{mul}(j;\tau)!}x^{l(\tau)}
=x(x+1)\cdots(x+k-1),
\label{step6}
\ea
that follows from the following computation,
\ba
&&\sum_{\tau\in P_{k}}\frac{k!}{\prod_{j=1}^{l(\tau)}n_{j}\prod_{j=1}^{k}\mathrm{mul}(j;\tau)!}x^{l(\tau)}\no\\
&=&\left.\frac{d^{k}}{dq^{k}}\biggl(1+\sum_{m=1}^{\infty}\sum_{\tau\in P_{m}}
\frac{1}{\prod_{j=1}^{l(\tau)}n_{j}\prod_{j=1}^{m}\mathrm{mul}(j;\tau)!}q^{m}x^{l(\tau)}\biggr)\right|_{q=0}\no\\                                                         \no\\
&=&\left.\frac{d^{k}}{dq^{k}}\biggl(\prod_{n=1}^{\infty}\exp(\frac{q^{n}}{n}x)\biggr)\right|_{q=0}\no\\
&=&\left.\frac{d^{k}}{dq^{k}}\biggl(\exp(\sum_{n=1}^{\infty}\frac{q^{n}}{n}x)\biggr)\right|_{q=0}\no\\
&=&\left.\frac{d^{k}}{dq^{k}}\biggl(\exp(-x\log(1-q))\biggr)\right|_{q=0}\no\\
&=&\left.\frac{d^{k}}{dq^{k}}(1-q)^{-x}\right|_{q=0}\no\\
&=&x(x+1)\cdots(x+k-1).
\ea
By combining (\ref{step4}), (\ref{step5}) and (\ref{step6}), we obtain, 
\ba
&&\sum_{\sigma  \in S_k}\mathrm{sgn}(\sigma)\biggl(\prod_{a=1}^k \frac{\partial}{\partial x_{a,\sigma(a)}} \biggr) |X|^{l+1}\no\\
&=&(l+1)(l+2)\cdots(l+k)|X|^{l}\no\\
&=&\frac{(l+k)!}{l!}|X|^{l}.
\ea
This completes the proof of the lemma.
\end{proof}

\begin{lemma}
\begin{align}
C(l+1,k)=\frac{(l+k)!}{l!}C(l,k). \label{eq:con2}
\end{align}
\label{lemma6}
\end{lemma}

\begin{proof} 
We start from Lemma \ref{clklem} applied to $C(l+1,k)$.  
\ba
C(l+1,k)=\sum_{M_{l+1}\in{\cal M}_{l+1}}\Biggl[\prod_{a,b=1}^k (m^{l+1}_{a,b}) ! \Biggr] \mul(M_{l+1})^2
\label{lstep1}
\ea
By combining the above equality with (\ref{lemma2}), we obtain,
\ba
C(l+1,k)&=&\sum_{M_{l+1}\in{\cal M}_{l+1}}\sum_{\sigma,\tau  \in S_k}\Biggl[\prod_{a,b=1}^k (m_{a,b}^{l+1})! \Biggr] \mathrm{sgn}(\sigma) \mathrm{sgn}(\tau) \no \\
&&\times \mul(M_{l+1}-R(\sigma)) \mul(M_{l+1}-R(\tau)).
\ea
We set $M_l=M_{l+1}-R(\sigma)$ and rewrite further the above equality.
\ba
C(l+1,k)&=&\sum_{M_{l}\in{\cal M}_{l}}\sum_{\sigma,\tau  \in S_k}\Biggl[\prod_{a,b=1}^k (m_{a,b}^{l}+\delta_{\sigma(a),b})! \Biggr]\mathrm{sgn}(\sigma) \mathrm{sgn}(\tau) \mul(M_{l}) \no \\
&&\times \mul(M_{l}+R(\sigma)-R(\tau))\no \\
&=&\sum_{M_{l}\in {\cal M}_{l}}\Biggl[\prod_{a,b=1}^k (m_{a,b}^{l})! \Biggr] \mul(M_{l}) \sum_{\sigma,\tau  \in S_k}\Biggl[\prod_{a=1}^k (m_{a,\sigma(a)}^{l}+1) \Biggr]\mathrm{sgn}(\sigma) \mathrm{sgn}(\tau)  \no \\
&&\times \mul(M_{l}+R(\sigma)-R(\tau)) .
\ea
Hence in order to prove the lemma, we only have to confirm the following equality:
\begin{align}
&\sum_{\sigma,\tau  \in S_k}\Biggl[\prod_{a=1}^k (m_{a,\sigma(a)}^{l}+1) \Biggr]\mathrm{sgn}(\sigma) \mathrm{sgn}(\tau) \mul(M_{l}+R(\sigma)-R(\tau)) =\frac{(k+l)!}{l!}\mul(M_{l}) .
\label{goal}
\end{align}
If $\sigma=\tau$, $\mathrm{sgn}(\sigma) \mathrm{sgn}(\tau) \mul(M_{l}+R(\sigma)-R(\tau))$ obviously equals $\mul(M_{l})$. 
Let us use the representation $M_{l}=\sum_{\sigma^{\prime}\in S_{k}}n_{\sigma^{\prime}}R(\sigma^{\prime})$.
If $\sigma\neq \tau$, we have the following representation:
\ba
M_{l}+R(\sigma)-R(\tau)=\sum_{\sigma^{\prime}\neq \sigma,\tau}n_{\sigma^{\prime}}R(\sigma^{\prime})
+(n(\sigma)+1)R(\sigma)+(n(\tau)-1)R(\tau).
\ea 
Applying (\ref{Ml}) carefully to these two cases, we obtain,
\begin{eqnarray}
&&\mathrm{sgn}(\sigma) \mathrm{sgn}(\tau) \mul(M_{l}+R(\sigma)-R(\tau)) \notag \\
&&=
\left \{ 
\begin{array}{c}
\displaystyle \sum_{\varphi((n_{\sigma^{\prime}})_{\sigma^{\prime}\in S_{k}})=M_{l}} \frac{l!}{\prod_{\sigma^{\prime} \in S_k} n_{\sigma^{\prime}}!} \prod_{\sigma^{\prime} \in S_k}\Bigl(\mathrm{sgn}(\sigma^{\prime})\Bigr) ^{n_{\sigma^{\prime}}} ~~(\sigma=\tau), \\
 \displaystyle \sum_{\varphi((n_{\sigma^{\prime}})_{\sigma^{\prime}\in S_{k}})=M_{l}}\frac{l!}{\prod_{\sigma^{\prime} \in S_k} n_{\sigma^{\prime}}!} \times \frac{n_{\tau}}{n_{\sigma}+1} \prod_{\sigma^{\prime} \in S_k}\Bigl(\mathrm{sgn}(\sigma^{\prime})\Bigr) ^{n_{\sigma^{\prime}}} ~~(\sigma \neq \tau).
\end{array} 
\right.
\end{eqnarray} 
Hence we can rewrite the l.h.s. of (\ref{goal}) as follows. 
\begin{align}
&\sum_{\sigma,\tau  \in S_k}\Bigl(\prod_{a=1}^k (m_{a,\sigma(a)}^{l}+1) \Bigr)\mathrm{sgn}(\sigma) \mathrm{sgn}(\tau) \mul(M_{l}+R(\sigma)-R(\tau)) \notag \\
&=\sum_{\sigma  \in S_k}\Bigl(\prod_{a=1}^k (m_{a,\sigma(a)}^{ l}+1) \Bigr) \Bigl \{ \displaystyle  \sum_{\varphi((n_{\sigma^{\prime}})_{\sigma^{\prime}\in S_{k}})=M_{l}}\frac{l!}{\prod_{\sigma^\prime \in S_k} n_{\sigma^\prime}!} \prod_{\sigma^\prime \in S_k}\Bigl(\mathrm{sgn}(\sigma^\prime)\Bigr) ^{n_{\sigma^\prime}} \notag \\
& +\sum_{\tau \neq \sigma} \sum_{\varphi((n_{\sigma^{\prime}})_{\sigma^{\prime}\in S_{k}})=M_{l}} \frac{l!}{\prod_{\sigma^\prime \in S_k} n_{\sigma^\prime}!} \times \frac{n_{\tau}}{n_{\sigma}+1} \prod_{\sigma^\prime \in S_k}\Bigl(\mathrm{sgn}(\sigma^\prime)\Bigr) ^{n_{\sigma^\prime}} \Bigr\} \no\\
&=\sum_{\sigma  \in S_k}\Bigl(\prod_{a=1}^k (m_{a,\sigma(a)}^{ l}+1) \Bigr) \Bigl \{ \displaystyle  \sum_{\varphi((n_{\sigma^{\prime}})_{\sigma^{\prime}\in S_{k}})=M_{l}}\frac{l!}{\prod_{\sigma^\prime \in S_k} n_{\sigma^\prime}!} \prod_{\sigma^\prime \in S_k}\Bigl(\mathrm{sgn}(\sigma^\prime)\Bigr) ^{n_{\sigma^\prime}} \notag \\
& + \sum_{\varphi((n_{\sigma^{\prime}})_{\sigma^{\prime}\in S_{k}})=M_{l}}\frac{l!}{\prod_{\sigma^\prime \in S_k} n_{\sigma^\prime}!}\times \frac{l-n_\sigma}{n_{\sigma}+1}  \prod_{\sigma^\prime \in S_k}\Bigl(\mathrm{sgn}(\sigma^\prime)\Bigr) ^{n_{\sigma^\prime}} \Bigr\} \no\\
&=\sum_{\sigma  \in S_k}\Bigl(\prod_{a=1}^k (m_{a,\sigma(a)}^{ l}+1) \Bigr)  
 \sum_{\varphi((n_{\sigma^{\prime}})_{\sigma^{\prime}\in S_{k}})=M_{l}} \frac{(l+1)!}{(n_\sigma+1) \prod_{\sigma^\prime \in S_k} n_{\sigma^\prime}!}  \prod_{\sigma^\prime \in S_k}\Bigl(\mathrm{sgn}(\sigma^\prime)\Bigr) ^{n_{\sigma^\prime}}. 
\label{ess0}
\end{align}
In going from the second expression to the third expression, we used the equality:
\ba
\sum_{\tau\neq\sigma}n_{\tau}=\sum_{\tau\in S_{k}}n_{\tau}-n_{\sigma}=l-n_{\sigma}.
\ea
At this stage, we explicitly compute the l.h.s. of  Lemma \ref{lemma5} $\sum_{\sigma  \in S_k}\mathrm{sgn}(\sigma)\Bigl(\prod_{a=1}^k \frac{\partial}{\partial x_{a, \sigma(a)}} \Bigr) |X|^{l+1}$ by using the expansion $ |X|^{l+1}=\sum_{M_{l+1}\in {\cal M}_{l+1}} \mul(M_{l+1}) \prod_{a,b=1}^{k}  x_{a,b}^{m^{l+1}_{a,b}}$. 
\begin{align}
&\sum_{\sigma  \in S_k}\mathrm{sgn}(\sigma)\Bigl(\prod_{a=1}^k \frac{\partial}{\partial x_{a \sigma(a)}} \Bigr) |X|^{l+1} \notag \\
&=\sum_{\sigma  \in S_k}\mathrm{sgn}(\sigma)\Bigl(\prod_{a=1}^k \frac{\partial}{\partial x_{a \sigma(a)}} \Bigr) \sum_{M_{l+1}\in {\cal M}_{l+1}} \sum_{\varphi((n_{\tau})_{\tau \in S_{k}})=M_{l+1}}\frac{(l+1)!}{\prod_{\tau \in S_k} n_{\tau}!} \Biggl[\prod_{\tau \in S_k}\Bigl(\mathrm{sgn}(\tau)\Bigr)^{n_\tau}\Biggr]   \prod_{a,b=1}^{k} x_{a,b}^{m_{a,b}^{l+1}} \notag \\
&=\sum_{\sigma  \in S_k} \sum_{M_{l+1}\in {\cal M}_{l+1}} \sum_{\varphi((n_{\tau})_{\tau \in S_{k}})=M_{l+1}}\frac{(l+1)!\mathrm{sgn}(\sigma)}{\prod_{\tau \in S_k} n_{\tau}!} \Biggl[\prod_{\tau \in S_k}\Bigl(\mathrm{sgn}(\tau)\Bigr)^{n_\tau} \Biggr] \Biggl[ \prod_{a=1}^k m_{a,\sigma(a)}^{l+1} x_{a,\sigma(a)}^{m_{a,\sigma(a)}^{l+1}-1}   \Biggr] \notag \\
&\times \prod_{a=1 }^{k}\prod_{\substack{b=1\\(b\neq\sigma(a))}}^k x_{a,b}^{m_{a,b}^{l+1}}.
\end{align}
In the last expression, $\varphi((n_{\tau})_{\tau \in S_{k}})=M_{l+1}$ that corresponds to non vanishing summand satisfies the condition 
$n_{\sigma}\geq 1$. Hence we can set $M_{l+1}=M_{l}+R(\sigma)$ with $\varphi((m_{\tau})_{\tau \in S_{k}})=M_{l}$. Then we can further rewrite 
the above expression,
\begin{align}
&=\sum_{\sigma  \in S_k}\sum_{M_{l}\in{\cal M}_{l}} \sum_{\varphi((m_{\tau})_{\tau \in S_{k}})=M_{l}}\frac{(l+1)!}{(m_\sigma+1)\prod_{\tau} m_{\tau}!} (\mathrm{sgn}(\sigma))^2\Biggl[ \prod_{\tau \in S_k}\Bigl(\mathrm{sgn}(\tau)\Bigr)^{m_\tau} \Biggr]  \notag \\
&\times \Biggr[\prod_{a=1}^k(m_{a,\sigma(a)}^{l}+1) x_{a,\sigma(a)}^{m_{a,\sigma(a)}^{l}}  \Biggr]
\prod_{a=1 }^{k}\prod_{\substack{b=1\\(b\neq\sigma(a))}}^k x_{a,b}^{m_{a,b}^{l}} \notag \\
&=\sum_{\sigma  \in S_k}\sum_{M_{l}\in{\cal M}_{l}}\sum_{ \varphi((m_{\tau})_{\tau \in S_{k}})=M_{l}}\frac{(l+1)!}{(m_\sigma+1)\prod_{\tau} m_{\tau}!} \Biggl[\prod_{\tau \in S_k}\Bigl(\mathrm{sgn}(\tau)\Bigr)^{m_\tau}  \Biggr] 
\Biggl[\prod_{a=1}^k(m_{a,\sigma(a)}^{l}+1)\Biggr]\prod_{a,b=1 }^{k} x_{a,b}^{m_{a,b}^{l}}. \label{proof2,2}
\end{align}
By combining, the above derivation with assertion of Lemma \ref{lemma5},
\ba
&&\sum_{\sigma  \in S_k}\mathrm{sgn}(\sigma)\Bigl(\prod_{a=1}^k \frac{\partial}{\partial x_{a \sigma(a)}} \Bigr) |X|^{l+1} \notag \\
&&=\sum_{M_{l}\in{\cal M}_{l}}\sum_{\varphi((m_{\tau})_{\tau \in S_{k}})=M_{l}} \sum_{\sigma  \in S_k} \frac{(l+1)!}{(m_\sigma+1)\prod_{\tau} m_{\tau}!} \Biggl[\prod_{\tau \in S_k}\Bigl(\mathrm{sgn}(\tau)\Bigr)^{m_\tau}  \Biggr] \Biggl[\prod_{a=1}^k(m_{a,\sigma(a)}^{l}+1)\Biggr]\prod_{a,b=1 }^{k} x_{a,b}^{m_{a,b}^{l}}\no\\
&&=\frac{(k+l)!}{l!}\sum_{M_{l}\in{\cal M}_{l}}\mul(M_{l}) \prod_{a,b=1}^{k}  x_{a,b}^{m^{l}_{a,b}}, 
\ea
we obtain, 
\ba
\sum_{ \substack{M_{l}=\\ \sum_{\tau} m_{\tau}R(\tau)} } \sum_{\sigma  \in S_k} \frac{(l+1)!}{(m_\sigma+1)\prod_{\tau} m_{\tau}!} \Biggl[\prod_{\tau \in S_k}\Bigl(\mathrm{sgn}(\tau)\Bigr)^{m_\tau}  \Biggr] \Biggl[\prod_{a=1}^k(m_{a,\sigma(a)}^{l}+1)\Biggr]=\frac{(k+l)!}{l!}\mul(M_{l}).\no\\
\label{essential} 
\ea
By comparing (\ref{ess0}) with (\ref{essential}), we reach the equality:
\ba
&&\sum_{\sigma,\tau  \in S_k}\Biggl[\prod_{a=1}^k (m_{a,\sigma(a)}^{l}+1) \Biggr]\mathrm{sgn}(\sigma) \mathrm{sgn}(\tau) \mul(M_{l}+R(\sigma)-R(\tau))=\frac{(k+l)!}{l!}\mul(M_{l}),
\ea
which completes the proof of the lemma.
\end{proof}
{\bf Proof of Theorem 3}\\
\begin{proof} 
$C(1,k)$ is calculated as follows.
\begin{align}
C(1,k)&= \sum_{\sigma \in S_k}\mathrm{sgn} (\sigma)\int D\psi  \psi_{1}^{1}\psi_{\bar 1}^{\overline{\sigma(1)}}\cdots \psi_{1}^{k}\psi_{\bar 1}^{\overline{ \sigma(k)}} \no \\
&=\sum_{\sigma \in S_k} (\mathrm{sgn} (\sigma) )^2\int D\psi \psi_{1}^{1}\psi_{\bar 1}^{\bar 1}\cdots \psi_{1}^{k}\psi_{\bar 1}^{\bar k}=\sum_{\sigma \in S_k} 1=k!.
\end{align}
Then successive use of  Lemma \ref{lemma6} leads us to,
\ba
C(l,k)=\frac{(l+k-1)!}{(l-1)!}\cdots\frac{(k+1)!}{1!}C(1,k)=\frac{\prod_{j=0}^{k+l-1}j!}{\prod_{j=0}^{k-1}j!\prod_{j=0}^{l-1}j!}.
\ea
Hence in the case of Theorem 3, the l.h.s. $\frac{ \prod_{j=0}^{k-1} j!}{ \prod_{j=N-k}^{N-1} j!}\int D\psi^\prime  \left(\det(\Phi^\prime)\right)^{N-k}$ equals $\frac{ \prod_{j=0}^{k-1} j!\prod_{j=0}^{N-k-1}j!}{ \prod_{j=0}^{N-1} j!}C(N-k,k)
=\frac{ \prod_{j=0}^{k-1} j!\prod_{j=0}^{N-k-1}j!}{ \prod_{j=0}^{N-1} j!}\cdot\frac{\prod_{j=0}^{N-1}j!}{\prod_{j=0}^{k-1}j!\prod_{j=0}^{N-k-1}j!}=1$.

\end{proof}

\end{document}

%% file: Eulernumber.tex
\subsubsection{Integration of Fields except for $\psi$}
 
Note that we have integrated out $z$ in the previous subsubsection. Next, we integrate out $\psi_{A}$, $H$ and $\chi$.
\begin{align}
&\int D\psi_A \exp( \mathrm{tr}(\psi^\dag_A \psi_A)) \notag \\
&=\int d{\psi_{A}}^1_1d{\psi_{A}}^{\bar 1}_{\bar 1} \cdots d{\psi_{A}}^k_1d{\psi_{A}}^{\bar k}_{\bar 1} \cdots d{\psi_{A}}^1_kd{\psi_{A}}^{\bar 1}_{\bar k} \cdots d{\psi_{A}}^k_kd{\psi_{A}}^{\bar k}_{\bar k} \prod_{i=1}^k \{(1+{\psi_{A}}^{\bar i}_{\bar 1}{\psi_{A}}^i_1)\cdots (1+{\psi_{A}}^{\bar i}_{\bar k}{\psi_{A}}^i_k)\} \notag \\
&=(-1)^{k^2}.\\
&\int DH \exp(-\sum^N_{s=1}\sum^k_{i=1}H^{\bar i}_{\bar s} H^i_s) \notag \\ 
&:=\int \prod_{s=1}^N (\frac{i}{2})^k dH^1_ s dH^{\bar 1}_{\bar s} \cdots dH^k_s dH^{\bar k}_{\bar s} \exp(-\sum^N_{s=1}\sum^k_{i=1}H^{\bar i}_{\bar s} H^i_s)
=\pi^{kN}.\\
&\int D\chi \exp(\sum^N_{s=1}\chi^{\bar i}_{\bar s}\bigl(\gamma\delta_{\bar i j}+i\phi_{\bar i j} \bigr)\chi^j_s) \notag \\
&:=\int \prod_{s=1}^N d\chi^1_s d\chi^{\bar 1}_{\bar s} \cdots d\chi^k_s d\chi^{\bar k}_{\bar s} \exp(\sum^N_{s=1}\chi^{\bar i}_{\bar s} \bigl(\gamma\delta_{\bar i j}+i\phi_{\bar i j} \bigr)\chi^j_s) \notag \\
&= (-1)^{kN} \int \prod_{s=1}^N d\chi^{\bar 1}_{\bar s} d\chi^1_s \cdots d\chi^{\bar k}_{\bar s} d\chi^k_s  \exp(\sum^N_{s=1}\chi^{\bar i}_{\bar s}\bigl(\gamma\delta_{\bar i j}+i\phi_{\bar i j} \bigr)\chi^j_s) \notag \\
&=\left( -1\right)^{kN} \left(\det(\gamma I_k+i\phi)\right)^N.
\end{align}
Then let us integrate out $\eta$. We set, 
\begin{align}
I_\eta&:=\int D\eta \exp \left( \frac{-1}{2} \sum_{i,j=1}^k \left\{ \psi^{\bar i}_{\bar j} -\psi^{ j}_{ i}  \right\}\eta_{\bar i j} \right)
\end{align}
We abbreviate $\left\{ \psi^{\bar i}_{\bar j} -\psi^{ j}_{ i}  \right\}$ as $\alpha^{\bar i j}$ and obtain, 
\begin{align}
I_\eta&=\int D\eta \left( \frac{-1}{2} \right)^{k^2} \prod_{i=1}^k \prod_{j=1}^k \alpha^{\bar i j} \eta_{\bar i j}
=\left( \frac{-1}{2} \right)^{k^2} (-1)^{\frac{k^2}{2}(k^2+1)} \int D\eta  \left[\prod_{i=1}^k \prod_{j=1}^k \eta_{\bar i j}\right]  \left[\prod_{i=1}^k \prod_{j=1}^k \alpha^{\bar i j} \right]\\
&=\left( \frac{-1}{2} \right)^{k^2} (-1)^{\frac{k^2}{2}(k^2+1)}  \left[\prod_{i=1}^k \prod_{j=1}^k \alpha^{\bar i j} \right].
\end{align}
We also abbreviate $\left\{ \psi^{\bar i}_{\bar j} +\psi^{ j}_{ i}  \right\}$ as $\beta^{\bar i j}$ and evaluate $I_\eta \omega$. Since 
\begin{align}
\left[\prod_{i=1}^k \prod_{j=1}^k \alpha^{\bar i j} \right]\left[\prod_{i=1}^k \prod_{j=1}^k \beta^{\bar i j} \right]
=(-1)^{\frac{k^2}{2}(k^2+1)}\left[\prod_{i=1}^k \prod_{j=1}^k \alpha^{\bar i j} \beta^{\bar i j}\right],
\end{align}
and $\alpha^{\bar i j} \beta^{\bar i j}=(\psi^{\bar i}_{\bar j} -\psi^{ j}_{ i})(\psi^{\bar i}_{\bar j} +\psi^{ j}_{ i})=-2\psi^{ j}_{ i}\psi^{\bar i}_{\bar j}$
, we obtain, 
\begin{align}
I_\eta \omega&=\left( \frac{-1}{2} \right)^{k^2} \left[\prod_{i=1}^k \prod_{j=1}^k \alpha^{\bar i j} \beta^{\bar i j}\right]
=\prod_{i=1}^k \prod_{j=1}^k \psi^{ j}_{ i}\psi^{\bar i}_{\bar j}
=\left(\prod_{i=1}^k  \psi^{ i}_{ i}\psi^{\bar i}_{\bar i} \right)\left(\prod_{i<j} \psi^{ j}_{ i}\psi^{\bar i}_{\bar j} \right)\left(\prod_{i>j} \psi^{ j}_{ i}\psi^{\bar i}_{\bar j} \right) \notag \\
&=\left(\prod_{i=1}^k  \psi^{ i}_{ i}\psi^{\bar i}_{\bar i} \right)\left(\prod_{i<j} \psi^{ j}_{ i}\psi^{\bar i}_{\bar j}  \psi^{ i}_{ j}\psi^{\bar j}_{\bar i} \right) 
=\left(\prod_{i=1}^k  \psi^{ i}_{ i}\psi^{\bar i}_{\bar i} \right)\left(\prod_{i<j} (-\psi^{ j}_{ i} \psi^{\bar j}_{\bar i} \psi^{ i}_{ j}\psi^{\bar i}_{\bar j}) \right) \notag \\
&=(-1)^{\frac{k(k-1)}{2}}\prod_{i=1}^k \prod_{j=1}^k \psi^{ j}_{ i} \psi^{\bar j}_{\bar i}.
\end{align}
Then the result of these integrations is given as follows. 
\begin{align}
Z_{MQ}&=\bt\int_{V_k(\mathbb{C}^N)} Dz D\psi D\phi D\bar \phi DA   \,\, (-1)^{k^2+\frac{k(k-1)}{2}}(-\pi )^{kN} \left(\det(\gamma I_k+i\phi)\right)^N \notag \\ 
&\exp(-\Bigl \{  \sum^N_{s=1}\psi^{ \bar i}_{\bar s} \bar \phi _{\bar i l} \psi^{  l}_s +i\mathrm{tr}(\phi \bar \phi)+\mathrm{tr}(A^\dag(\gamma A+[i\phi,A]))\Bigr \}) \prod_{i=1}^k \prod_{j=1}^k \psi^{ j}_{ i} \psi^{\bar j}_{\bar i}. 
\end{align}
In order to integrate out $\phi$ and  $\bar \phi$, we we decompose complex fields $\phi$, $\bar\phi$ and $\psi$ into real parts and imaginary parts:  
\begin{align}
&\phi_{\bar i j}=\phi^R_{i j}+i \phi^I_{i j} ,&
&\phi_{ i \bar j}=\phi^R_{i j}-i \phi^I_{i j}. &\\
&\bar \phi_{\bar i j}=\bar \phi^R_{i j}+i\bar \phi^I_{i j} ,&
&\bar \phi_{ i \bar j}=\bar \phi^R_{i j}-i \bar \phi^I_{i j}. &\\
&\psi^{ i}_s=\psi^{i}_{Rs}+i\psi^{ i}_{Is}, &
& \psi^{ \bar i}_{\bar s}=\psi^{ i}_{Rs}-i\psi^{  i}_{I s},&
\end{align}
and define integration measures for $\bar\phi$ and $\phi$ as,
\begin{align}
D\bar\phi&:=\prod_{i=1}^k d\bar\phi^{R}_{\bar i i} \prod_{j=i+1}^k d\bar\phi^R_{ i j} d\bar\phi^I_{i j},\\
D\phi&:=\prod_{i=1}^k d\phi^{R}_{\bar i i} \prod_{j=i+1}^k d\phi^R_{ i j} d\phi^I_{i j}.
\end{align}
Note that $\bar\phi^I_{ii}=\phi^I_{ii}=0$ and $\phi^R_{ij}=\phi^R_{j i}, \;\bar\phi^R_{ij}=\bar\phi^R_{j i},\; \phi^I_{i j}=-\phi^I_{j i}, \;\phi^I_{i j}=-\phi^I_{j i} (i\neq j)$ since both $\bar\phi$ and $\phi$ are both Hermitian matrices.   
Then integration of $\phi$ and $\bar\phi$ results in the following lemma:
\begin{lemma}
Let $\lambda_i (i=1,2,\cdots,k)$ be eigenvalues of the Hermite matrix $\phi$. Then we have,
\begin{align}
&\int D\bar \phi \exp\left (-i\mathrm{tr}(\phi \bar \phi)-\sum^N_{s=1} \psi^{ \bar i}_{\bar s} \bar \phi _{\bar i l} \psi^{ l}_s \right)
\notag \\
 &=\frac{(2\pi)^{k^2}}{2^{k(k-1)}}\left\{\prod_{i=1}^k \delta \left( \phi_{\bar i i}   + 2\sum^N_{s=1} \psi^{ i}_{Rs}\psi^{ i}_{Is}\right) \right\} \notag \\
&\times \Biggl \{ \prod_{i=1}^{k-1} \prod_{j=i+1}^k \delta \left( \phi^R_{ j i}+\sum_{s=1}^N(\psi^{  i}_{Rs}\psi^{ j}_{Is}-\psi^{  i}_{I s}\psi^{ j}_{Rs}) \right) \notag \\
&\times \delta \left( \phi^I_{ j i} -\sum^N_{s=1} (\psi^{  i}_{Rs}\psi^{ j}_{Rs}+\psi^{  i}_{I s}\psi^{ j}_{Is}) \right) \Biggr \}. \\
 &\int DA\exp\left(-\mathrm{tr}(\gamma A^\dag A+A^\dag [i\phi,A]) \right)
=\frac{\pi^{k^2}}{\gamma^{k}\prod_{l<j}(\gamma^2-(i\lambda_l-i\lambda_j)^2)}
\end{align}
\end{lemma}
This follows from straightforward computation, and we leave the proof to readers as exercises.  
By using the above lemma, we obtain,
\begin{align}
&Z_{MQ}=\bt\int D\psi^\prime Dz  D\phi \frac{(-1)^{\frac{k}{2}(k-1)}(2\pi)^{k^2}  (-\pi)^{k^2+kN}\left(\det(\gamma I_k+i\phi)\right)^N }{2^{k(k-1)}\gamma^{k}  \prod_{l>j} \left(\gamma^2-(i\lambda_l-i\lambda_j)^2 \right)}   \,\, \notag \\
&\left\{\prod_{i=1}^k \delta \left( \phi_{\bar i i}  + 2\sum^N_{s=1} \psi^{ i}_{Rs}\psi^{ i}_{Is} \right) \right\}  \Biggl \{ \prod_{i=1}^{k-1} \prod_{j=i+1}^k \delta \left( \phi^R_{ j i}+\sum_{s=1}^N(\psi^{  i}_{Rs}\psi^{ j}_{Is}-\psi^{  i}_{I s}\psi^{ j}_{Rs}) \right) \notag \\
&\times \delta \left( \phi^I_{ j i} -\sum^N_{s=1} (\psi^{ i}_{Rs}\psi^{ j}_{Rs}+\psi^{ i}_{I s}\psi^{ j}_{Is}) \right) \Biggr \}  \prod_{i=1}^k \prod_{j=1}^k \psi^{ j}_{ i} \psi^{\bar j}_{\bar i}
\label{med1}
\end{align}
Then, integration of $\phi$ results in replacement of $\phi$ fields by the following composite of $\psi$ fields.
\begin{align}
\phi_{\bar i i} &\to-2\sum^N_{s=1} \psi^{  i}_{Rs}\psi^{ i}_{Is}
=-\sum^N_{s=1}( \psi^{  i}_{Rs}\psi^{ i}_{Is} -\psi^{ i}_{Is} \psi^{  i}_{Rs}) \notag \\
&=-i\sum^N_{s=1}( \psi^{  i}_{Rs}+i\psi^{ i}_{Is})(\psi^{ i}_{Rs}-i\psi^{ i}_{Is} )
=- i\sum_{s=1}^N \psi^{ i}_{ s} \psi^{ \bar i}_{\bar s} , 
\end{align}
\begin{align}
\phi_{\bar j i}=\phi^R_{j i}+i \phi^I_{j i} &\to  \sum_{s=1}^N \{-(\psi^{ i}_{Rs} \psi^{j}_{Is} - \psi^{ i}_{Is} \psi^{ j}_{R s})+i(\psi^{i}_{Rs} \psi^{ j}_{Rs} +\psi^{ i}_{Is} \psi^{j}_{Is})\} \notag \\
&=i \sum_{s=1}^N \{\psi^{ i}_{Rs} \psi^{ j}_{Rs} +\psi^{i}_{Is} \psi^{ j}_{Is}+i(\psi^{ i}_{Rs} \psi^{j}_{Is} - \psi^{ i}_{Is} \psi^{ j}_{R s}) \\
&=i\sum_{s=1}^N \psi^{ \bar i}_{s} \psi^{ j}_{s}
=-i\sum_{s=1}^N \psi^{ j}_{s} \psi^{ \bar i}_{s}.
\end{align} 
At this stage, let us introduce the following Hermitian matrix:
 \begin{align}
& \Phi_s:=\left(
\begin{array}{ccc}
\psi^{ 1}_{s} \psi^{ \bar 1}_{\bar s} & \cdots &\psi^{ 1}_{s}\psi^{\bar k}_{\bar s}  \\
\vdots & \ddots & \vdots \\
\psi^{k}_{s}  \psi^{\bar 1}_{\bar s}& \cdots &\psi^{ k}_{s}  \psi^{\bar k}_{\bar s} 
\end{array}\right),\;\;\;\;\Phi:=\sum_{s=1}^N \Phi_s.
\end{align}
Then the result of integration of $\phi$ is summarized by replacement of $\phi$ by $-i\Phi$. 
Let $\lambda_{i}^\prime\;\;(i=1,\cdots,k)$ be  eigenvalues of the  $\Phi$.
Then integration of $\bar\phi$ and $\phi$ results in,  
\begin{align}
&Z_{MQ}=\bt\mathrm{vol}(V_k(\mathbb{C}^N))\int D\psi  \frac{(-1)^{\frac{k}{2}(k-1)}2^k \pi^{k^2} (-\pi) ^{k^2+kN}\left(\det(\gamma I_k+\Phi)\right)^N }{ \gamma^{k} \prod_{l>j} \left(\gamma^2-(\lambda^\prime_l-\lambda^\prime_j)^2 \right)}  \prod_{i=1}^k \prod_{j=1}^k \psi^{ j}_{ i} \psi^{\bar j}_{\bar i}.
\end{align}
On the other hand, we obtain from (\ref{volst}),
\begin{align}
\bt\mathrm{vol}(V_k(\mathbb{C}^N))(-1)^{\frac{k}{2}(k-1)}2^k \pi^{k^2} (-\pi) ^{k^2+kN}&=\frac{ \prod_{j=0}^{k-1} j!}{ \prod_{j=N-k}^{N-1} j!},
\end{align}
and reach the final expression of $Z_{MQ}$ in this section:
\begin{align}
&Z_{MQ}=\frac{ \prod_{j=0}^{k-1} j!}{ \prod_{j=N-k}^{N-1} j!}\int D\psi  \frac{\left(\det(\gamma I_k+\Phi)\right)^N }{ \gamma^{k} \prod_{l>j} \left(\gamma^{2}-(\lambda^\prime_l-\lambda^\prime_j)^2 \right)}  \prod_{i=1}^k \prod_{j=1}^k \psi^{ j}_{ i} \psi^{\bar j}_{\bar i}.  
\label{euler}
\end{align}